\documentclass[a4paper, 12pt]{article}

\usepackage[margin=2.4cm]{geometry}  % page formatting
\usepackage{amssymb, amsfonts, amsmath, amsthm, amsxtra}  % maths tools
\usepackage{thmtools, thm-restate}  % improved thm declarations, restate theorems
\usepackage{graphics, graphicx}  % pictures
\usepackage{old-arrows}  % old arrows
\usepackage{tabularx}  % width of columns
\usepackage{enumerate}  % list customization
\usepackage{authblk}  % authors formatting
\usepackage{hyperref}  % links coloring
\usepackage{xcolor}  % define colors
\usepackage{titlesec}  % title style change
\usepackage{tcolorbox}  % problem environments
\usepackage[lined, ruled, vlined, algo2e]{algorithm2e}  % algorithm appendices

\definecolor{clouds}{RGB}{236, 240, 241}  % grey
\definecolor{electron}{RGB}{9, 132, 227}  % blue
\definecolor{alizarin}{RGB}{231, 76, 60}  % red
\definecolor{asbestos}{RGB}{127, 140, 141}  % grey

\newcommand{\cc}[1]{\mathcal{#1}}  % calligraphic
\newcommand{\cb}[1]{\mathbb{#1}}  % doubled (IN, IR, ...)
\newcommand{\csf}[1]{\textsf{#1}}  % sans serif
\newcommand{\csmc}[1]{\textsc{#1}}  % small caps
\newcommand{\ctt}[1]{\texttt{#1}}  % text tt

\hypersetup{
    colorlinks=true,
    citecolor=alizarin,
    linkcolor=asbestos,
    urlcolor=alizarin,
}

\newcommand{\card}[1]{\vert #1 \vert}  % size of a set
\newcommand{\U}{R}  % relation scheme
\newcommand{\C}{\cc{C}}  % context
\newcommand{\IS}{\Sigma}  % implicational system
\newcommand{\Sb}{\text{\csf{Sub}}_{\land}^1}  % topped meet-sublattice
\newcommand{\Lt}{\cc{L}}  % lattice
\newcommand{\dom}{\csf{dom}}  % domain
\newcommand{\Pm}{\csf{Pr}}  % prime
\newcommand{\coPm}{\csf{CPr}} % co-prime.
\newcommand{\cl}{\phi}  % closure operator
\newcommand{\F}{\textbf{F}}  % functional dependencies
\newcommand{\NP}{\csf{NP}}  % NP
\newcommand{\DAV}[1]{\dom_{abs}(#1)}  % abstract domain

\DeclareMathOperator{\imp}{\rightarrow}  % implication arrow
\DeclareMathOperator{\eqv}{\Longleftrightarrow}  % equivalent
\DeclareMathOperator{\mt}{\land}  % meet
\DeclareMathOperator{\jn}{\lor}  % join
\DeclareMathOperator{\Mt}{\bigwedge}  % subset meet
\DeclareMathOperator{\Jn}{\bigvee}  % subset join

\declaretheorem[name=Theorem]{theorem}
\declaretheorem[name=Lemma]{lemma}

\declaretheorem[name=Proposition]{proposition}
\declaretheorem[name=Definition, style=definition]{definition}
\declaretheorem[name=Example, style=definition]{example}
\declaretheorem[name=Remark, style=remark]{remark}

\declaretheorem[name=Theorem, numbered=no]{theorem*}  % unnumbered theorem
\declaretheorem[name=Lemma, numbered=no]{lemma*}  % unnumbered lemma
\declaretheorem[name=Proposition, numbered=no]{proposition*}  % unnumbered proposition

\tcbuselibrary{many}

\newtcolorbox{problem}[1]{
colframe=clouds,
titlerule style=\color{black},
colback=white,
fonttitle=\color{black},
arc=1mm,
enhanced,
attach boxed title to top left={xshift=0.5cm, yshift=-3.7mm},
boxed title style={colframe=white},
title=\csmc{#1},
colbacktitle=white,
}

\newcommand{\PbState}[2]{
\begin{tabular}{l p{0.7\textwidth}}
	\textit{Input: \hspace{0.7em}} & #1 \\ %[0.2em]
	\textit{Question: \hspace{0.7em}} & #2 \\	
\end{tabular}
}

\newcommand{\Problem}[3]{
\begin{problem}{#1}
	\PbState{#2}{#3}
\end{problem}
}

\title{Towards declarative comparabilities: application to functional dependencies}

\author[1]{Lhouari Nourine}
\author[2]{Jean-Marc Petit}
\author[3]{Simon Vilmin}

\affil[1]{Universit\'e Clermont-Auvergne, CNRS, Mines de Saint-\'Etienne, 
Clermont-Auvergne-INP, LIMOS, 63000 Clermont-Ferrand, France.}
\affil[2]{INSA Lyon, CNRS, UCBL, Centrale Lyon, Univ Lyon 2, LIRIS,
UMR5205, F-69621 Villeurbanne, France}
\affil[3]{Aix-Marseille Universit\'e, CNRS, LIS, Marseille, France.}

\begin{document}

\maketitle

\begin{abstract}
In real life, data are often of poor quality as a result, for 
instance, of 
uncertainty, mismeasurements, missing values or bad inputs.
This issue hampers an implicit yet crucial operation of every database management system: 
equality testing.
Indeed, equality is, in the end, a context-dependent operation with a plethora of 
interpretations.
In practice, the treatment of different types of equality is left to programmers, who 
have to struggle with those interpretations in their code.
We propose a new lattice-based declarative framework to address this problem.
It allows specification of numerous semantics for equality at a high level of abstraction.
To go beyond tuple equality, we study functional dependencies (FDs) in the light of our 
framework.
First, we define abstract FDs, generalizing classical FDs.
These lead to the consideration of particular interpretations of equality: realities.
Building upon realities and possible/certain answers, we introduce possible/certain FDs 
and give some related complexity results.
\end{abstract}

\section{Introduction}
\label{sec:introduction}

Given two values $u$ and $v$, how can we decide whether or not $u$ and $v$ are 
equal?
The question is simple to ask, but much harder to answer as there exist countless 
possible interpretations of equality.
Moreover, it is an implicit but mainstream operation in every database management system 
(DBMS) as well as in most database applications among which query answering 
\cite{abiteboul1995foundations, amendola2018explainable, greco2014certain}, data 
integration \cite{lenzerini2002data}, inconsistent databases \cite{bertossi2011database}, 
probabilistic data \cite{suciu2011probabilistic}, data profiling \cite{abedjan2018data}, 
and database 
design \cite{link2019relational}.

At the same time, the quality of real life data is (very) often quite poor 
\cite{batini2009methodologies}.
This is due to issues such as uncertainty, measurement errors, bad input or missing 
values to mention but a few.
In those contexts, equality testing becomes even more crucial for dealing with data 
quality issues. 

Finally, domain experts are usually the only ones who can specify what 
\textit{``equality''} truly means on their data.
As a consequence, the exact meaning of equality turns out to be plural and 
context-dependent, as many interpretations are indeed possible. 
In practice, implementation is left to DBMS programmers, who have to struggle with 
all those different interpretations of equality in their (SQL) code. 
%This is neither scalable nor efficient.

To cope with this issue, we introduce a new lattice-based declarative framework for 
relational databases.
It allows specification of different meanings of equality at a high level of abstraction.
As a consequence, we consider equality testing to be a first-class citizen.
The notion of comparabilities turns out to play a major role in our framework, defined at 
attribute level by \emph{comparability functions} and \emph{abstract lattices}. 
Different \emph{interpretations} are then possible as per the usual logic, making it 
possible to decide whether two values (or two tuples) are indeed equal. 
%This framework can be applied to many database problems, such as SQL query 
%answering, tuple de-duplication or record linkage for example.

One of the motivation for our work draws from a collaboration with an industrial partner specialized in cold chain, refrigereation and conditioning (\href{https://www.cemafroid.fr/index-en.htm}{CEMAFROID}).
The aim was to predict the ageing of refrigerating transport vehicles from their data \cite{LeGuilly21}.
Within this collaboration, a finer consideration of equality was a key notion to select relevant data in SQL in order to get better results.
Continuing this collaboration, we used the framework we introduce in this paper (with comparability functions, abstract values and lattices) to model 20 important attributes at a high level of abstraction.
A (preliminary) prototype has been implemented on top of PostgreSQL with two main components: first, a concrete language  on top of the DDL of PosgreSQL and second,  a SQL fragment to express simple selection queries with equality.

In this paper, we focus instead on its impact on functional dependencies (FDs). 
We define \emph{abstract FDs} on top of the lattice-based framework.
Abstract FDs lead to the consideration of particular interpretations, called realities.
Inspired by possible/certain query answering in databases \cite{libkin2016sql}, as well 
as weak/strong or possible/certain FDs \cite{al2022strongly, levene1998axiomatisation}, 
we use realities to define similar notions on functional dependencies, namely possible 
and certain FDs.
We give some complexity results about possible/certain FDs.

To deal with missing or uncertain informations, many works have introduced new 
types of dependencies (see e.g. 
\cite{baixeries2018characterizing, bertossi2013data, caruccio2015relaxed, 
link2019relational, ng2001extension}).
In these approaches, the comparison of two values is a one-step process which returns a 
truth value of the underlying logic (usually the classic binary logic or the real 
interval $[0, 1]$).
Our framework however replaces the operation of testing the equality of two values by a 
two-step process: comparison and interpretation.
Moreover, it does not modify the input data, contrary to common approaches for dealing 
with missing information \cite{al2022strongly, kohler2016possible, 
kohler2018sql, levene1998axiomatisation, lien1982equivalence}.
%The comparison step gives birth to dependencies based on lattice implications 
%\cite{day1992lattice} which lead to define particular interpretations.
Consequently, our framework supplies experts with the possibility to declare meaningful 
comparability functions and different semantics for equality separately, without having 
to modify their data.

Another well-studied way of dealing with uncertainty is to extend the classical 
relational approach to a new system based on more elaborated logics such as fuzzy 
relational systems (see e.g. \cite{belohlavek2018relational, bosc1999database, 
jevzkova2017fuzzy, prade1984generalizing, yahia1999extension}).
Nevertheless, unlike fuzzy systems, our framework relies on the classical relational 
model and its underlying binary logic.
Thus, the dependencies we consider are different in nature from the generalization of the 
classical FDs to the fuzzy set up.
Moreover, since our framework is based on the usual relational model, it can come on top 
of any DBMS and does not require the implementation of a whole new system.

\paragraph{Contributions.}
We summarize the main contributions of our paper:
\begin{itemize}
\item We introduce a declarative framework at database scheme level to consider 
different interpretations of equalities as first-class citizens
It is based on three concepts applied attribute-wise: comparability functions, 
abstract lattices, and interpretations.

\item We apply our framework to functional dependencies:
\begin{itemize}
\item We show that any relation can be associated with an \emph{abstract lattice}, thus 
giving a semantic for \emph{abstract functional dependencies}.
We prove that, under certain circumstances, an interpretation 
preserves the validity of classical FDs for any abstract lattice. 
These interpretations are called realities and strong realities.

%\item We show that each relation can be described by \emph{abstract functional 
%dependencies}, based on \emph{abstract lattices}.
%These lead to the study of particular interpretations called \emph{(strong) realities}.
%Especially, we give a lattice-theoretic characterization of realities.

\item We investigate the plausibility of a given functional dependency: whether it is 
\emph{possible}, i.e., there exists at least one reality for which the FD holds, or whether 
it is \emph{certain}, i.e., holds for every reality. 
We also discuss this problem for the case of strong realities and show that 
deciding strong possibility is $\NP$-complete.
\end{itemize}
\end{itemize}

\paragraph{Paper organization.}
In Section \ref{sec:nutshell}, we present our framework informally on a running example.
Section \ref{sec:preliminaries} recalls standard notions from database and lattice 
theory.
We formally introduce our lattice-based framework in Section \ref{sec:framework}.
In Section \ref{sec:fd}, we study functional dependencies in our framework.
In Section \ref{sec:related}, we discuss related works on handling incomplete and unclean 
information in databases.
At last, Section \ref{sec:conclusion} concludes with open questions for further research.

\section{Framework in a nutshell}
\label{sec:nutshell}

In this section, we describe our framework using the following running example.
\begin{example}(Running example)
A doctor is interested in the triglyceride level (mmol/L), the gender
(M/F), and the waist size (cm) of her patients.
We denote these attributes by $A$, $B$, $C$ respectively.
%We put $\U = \{A, B, C\}$ as the relation scheme of $\textrm{Patients}$.
A sample is represented in Table \ref{tab:medical-data} as the relation 
$\mathrm{Patients}$.
It contains a missing value, outliers, and some very similar values.
\begin{table}[ht!]
\centering
\includegraphics[scale=1.0, page=1]{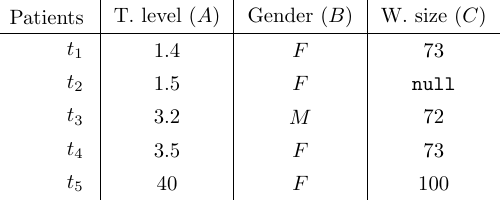}%
\caption{Running example on medical data}
\label{tab:medical-data}	

\end{table}
\end{example}

Each attribute is first equipped with a \emph{comparability function} 
that maps every pair of values in the attribute domain to an \emph{abstract value} in a 
fine-grained similarity scale, ordered in an \emph{abstract lattice}.
A comparability function must be associative.
% \remove{ and the comparison of a value with itself 
% (but for the \ctt{null} value) rises the highest possible abstract value.}
Both the comparability function and the abstract lattice can be given at database design 
time. 

\begin{example} (Continued)
Consider the waist size.
Based on her background knowledge, the doctor supplies conditions to compare 
different waist sizes.
First, she provides the main comparability values (abstract values), namely 
\textit{close}, \textit{unknown}, and \textit{distant}.
Then, she uses these abstract values to define an appropriate way to compare two 
waist sizes ($x$ and $y$ are waist sizes):
\begin{itemize}
	\item $\textit{close} \text{ if } x = y \neq \ctt{null} \text{ or } x, y 
	\in [70, 80]$,
	\item $\textit{unknown} \text{ if } u = \ctt{null} \text{ or } v = 
	\ctt{null} $,
	\item $\textit{distant} \text{ otherwise.}$
\end{itemize}
Eventually, she has to order the abstract values according to the principle 
\textit{``the higher the closer''}.
In the case of waist size, she ends up with the abstract lattice depicted in 
Figure \ref{fig:abstract-lattices} by the lattice $\Lt_C$.
For example, $72$ and $73$ are \textit{equal} while $73$ and $100$ are \textit{distant}. 
\begin{figure}[t]
	\centering 
	\includegraphics[page=2, scale=1.0]{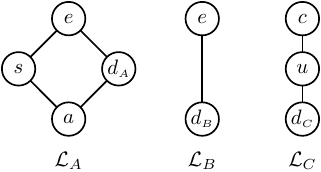}%
	\caption{Abstract lattices}
	\label{fig:abstract-lattices}
\end{figure}
More generally, we assume that the following comparability functions are defined, along 
with their associated abstract lattices (see Figure \ref{fig:abstract-lattices}).
\[
f_A(x, y) = \begin{cases}
	\textit{equal} & \text{ if } x = y \text{ or } x, y \in [0, 2[ \\
	\textit{second class} & \text{ if } x, y \in [2, 5[, x \neq y \\
	\textit{distributed} & \text{ if } (x, y) \text{ or } (y, x) \in [0, 2[ \times 
	[2, 5[ \\
	\textit{abnormal} & \text{ otherwise.}
\end{cases}
\]
\[
f_B(x, y) = \begin{cases}
	\textit{equal} & \text{ if } x = y \\
	\textit{different} & \text{ if } x \neq y. 
\end{cases}
\]
\[
f_C(x, y) = \begin{cases}
	\textit{close} & \text{ if } x = y \neq \ctt{null} \text{ or } x, y \in 
	[70, 80]\\
	\textit{unknown} & \text{ if } u = \ctt{null} \text{ or } v = \ctt{null} \\
	\textit{distant} & \text{ otherwise.}
\end{cases}
\]

We use abbreviations for each value: $e$ stands for \textit{equal}, $s$ for 
\textit{second class}, $a$ for \textit{abnormal}, $u$ for \textit{unknown}, $c$ for 
\textit{close}, and $d_A$ for \textit{distributed}, $d_B$ for \textit{different} and $d_C$ for \textit{distant}.
\end{example}

Remark that all three comparability functions extend classical equality since every value, except \ctt{null}, compared with itself gives the highest possible abstract value.
Then, a comparison of two tuples with this framework yields an \emph{abstract tuple}.
This is a precise account of their agreement on each attribute.

%Thus, whenever two tuples are compared in the framework, their comparison gives rise 
%to so-called an \emph{abstract 
	%	tuple}, i.e., a precise account of the truth of their agreement on each component.
%This comparison, obtained according to the background knowledge, naturally extends  
%equality. (BOF, a montrer justement)

\begin{example} %\label{ex:abstract-tuple}
(Continued) Consider the tuples $t_1, t_2$ of our running example.
Let $f_{\U}$ be the comparability function consisting of the attribute-wise 
application of $f_A, f_B$ and $f_C$.
The comparison of $t_1$ and $t_2$ is $f_{\U}(t_1, t_2) = \langle f_A(1.4, \allowbreak 
1.5), f_B(F, F), $ $f_C(73, \ctt{null}) \rangle = \langle e, e, u \rangle$.
\end{example}

Abstract tuples may not be intuitive enough to be used by practitioners.
Thus, the next step is to use \emph{interpretations} to decide whether or not an abstract 
value is considered to be equality.
At attribute level, an interpretation maps an abstract value to $0$ or $1$.
As usual, $0$ means difference, and $1$ means equality.
We further require that an interpretation is increasing: the higher the abstract value 
in the lattice, the more it gets close to equality. 
Thus, an interpretation is a semantic for equality on a attribute.
At the scheme level, an interpretation states whether or not two tuples are considered 
equal.

\begin{example}(Continued)
The doctor gives several hypotheses on the meaning of equality with respect to her 
data.
These lead to three possible interpretations $g_1, g_2, g_3$ which are represented in 
Figure \ref{fig:interpretations}.
Above the dotted lines, the abstract values are interpreted as $1$, and $0$ otherwise.
With $g_1$ and $g_2$, $t_1$ and $t_2$ are not equal, while they are with $g_3$ (we have 
$g_3( \langle e, e, u \rangle) = \langle 1, 1, 1 \rangle$).

\begin{figure}[ht!]
	\centering 
	\includegraphics[scale=0.8, page=3]{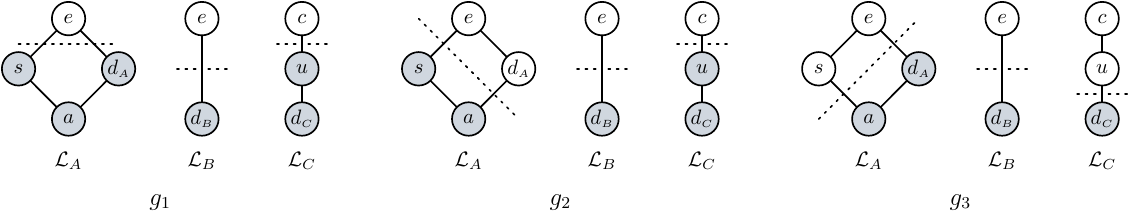}%
	\caption{Three possible interpretations of similarity values: $1$ 
		above the dotted line, $0$ below.}
	\label{fig:interpretations}
\end{figure}
	
\end{example}

Quite clearly, this framework provides domain experts with an opportunity to specify at a 
high level of abstraction different types of equality on top of database schemes, 
independently from the underlying applications and without updating the data.
We summarize our framework with the pipeline presented in 
Figure \ref{fig:framework}, applied to the tuples $t_1, t_4$ of $\textrm{Patients}$.

\begin{figure}[ht!]
	\centering 
	\includegraphics[scale=0.85, page=4]{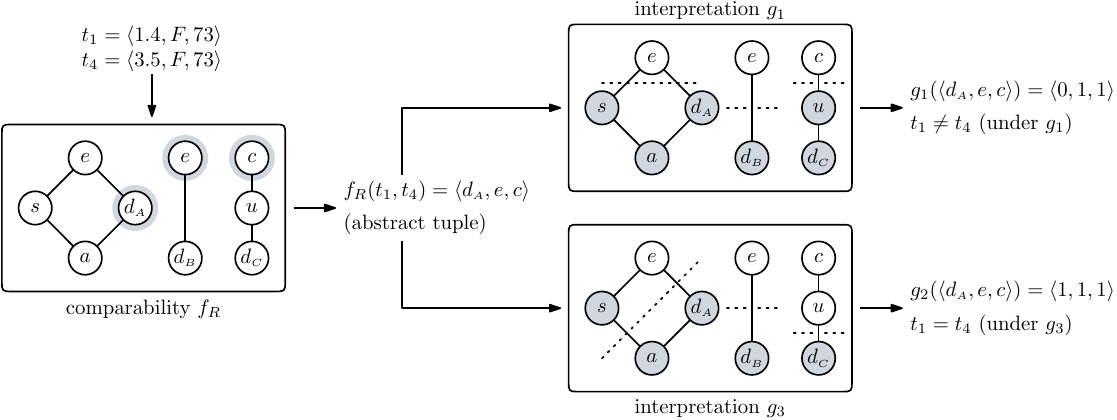}%
	\caption{Pipeline of the framework on $t_1$, $t_4$.}
	\label{fig:framework}
\end{figure}

\paragraph{Application to functional dependencies.}
Usually, a functional dependency $X \imp Y$ holds in a given relation if two tuples that 
are equal on $X$ are also equal on $Y$. 
This definition smoothly adapts to our framework using interpretations.
More precisely, $X \imp Y$ holds with respect to a given interpretation $g$ if two tuples 
that are considered equal on $X$ with respect to $g$, are also considered equal on $Y$ 
with respect to $g$.

From a more theoretical point of view, we show that a relation satisfies a more general 
type of dependencies based on abstract tuples, independently from any interpretation.
These dependencies are called \emph{abstract functional dependencies}.
They are lattice implications \cite{day1992lattice} of the form $x \imp y$, where $x$ and 
$y$ are abstract tuples drawn from the product of abstract lattices. 
In the manner of classical functional dependencies and agree sets 
\cite{beeri1984structure}, the abstract functional dependencies satisfied by a relation 
capture the abstract knowledge of this relation, being a collection of abstract tuples 
ordered in a new abstract lattice.
In this context, an interpretation turns the abstract lattice into a set system over the 
relation scheme.
This separates the interpretations: some of them preserve the semantic of 
classical FDs, i.e. the abstract lattice is interpreted as a closure system 
\cite{davey2002introduction}, and some do not.
We call \emph{realities} these semantic preserving interpretations.
It turns out that realities are precisely the interpretations that are lattice 
meet-homomorphisms \cite{davey2002introduction}.
As a side effect, the equality in realities can be represented by a unique minimum 
element in each attribute's abstract lattice.
We also study \emph{strong realities}, the interpretations that are lattice homomorphisms.

Building upon realities, we introduce the notion of plausibility 
(certain/possible) of a given FD, inspired from query answering 
\cite{amendola2018explainable, bertossi2011database, libkin2016sql} and weak/strong or 
possible/certain FDs \cite{levene1998axiomatisation, al2022strongly}.
More precisely, the existence of numerous realities leads to cases of equality semantic 
where the FD holds, and some others where it does not.
In this case, the FD is \emph{possible}.
In practice, this may be the case if, for instance, the relation suffers from unclean 
data: a FD 
should theoretically hold with respect to background knowledge, but due to 
mismeasurements or imprecisions, it fails to do so.
Here, weakening/strengthening the meaning of equality according to expert knowledge could 
easily fix this problem, by simply applying the \textit{``right''} reality with 
respect to the background knowledge.
After possibility, it is also natural to wonder for a given FD if there is a 
chance that it holds in every reality, in which case we say that it is \emph{certain}.
It is worth noticing that the data are left unchanged in our framework.
Hence, the terms \textit{``possible''} and \textit{``certain''} differ from the more 
usual point of view where a relation with \ctt{null} values is subject to numerous 
completions.

\begin{example}
(Continued)
The doctor knows from her background knowledge that the triglyceride level is 
determined by waist size and gender.
This can be modeled by the functional dependency $BC \imp A$.
Among the realities given in Figure \ref{fig:interpretations}, $g_2$ is the only one 
for which $BC \imp A$ holds.
The tuples $t_1$, $t_4$ represent a counter-example to $BC \imp A$ for $g_1$ and 
$g_3$.
Realities supply an abstract support to practitioners to help them understand under which 
conditions their background knowledge, modeled as a functional dependency, turns out to 
be indeed true or false in their data.
\end{example}

\section{Preliminaries}
\label{sec:preliminaries}

In this paper, we only consider finite structures. 
First, we recall basic definitions about relational databases \cite{maier1983theory, 
mannila1992design}.
Then, we move to lattice theoretical concepts \cite{davey2002introduction, 
gratzer2011lattice}.

\paragraph{Relational model} 
A \emph{relation scheme} $\U$ is a (finite) set of \emph{attributes} $A_1, \dots, 
A_n$, for $n \in \cb{N}$. 
Its powerset is written $2^{\U}$.
Each attribute $A_i$ has values in some set $\dom(A_i)$ called the 
\emph{domain} of $A_i$.
We consider a unique null value, \ctt{null}, which can be part of every attribute domain.
The domain of $\U$ is the union of attributes domains, that is $\dom(\U) = \bigcup_{A \in 
\U} \dom(A)$.
%For a given $n U \in \cb{N}$, we denote by $\dom(\U)^n$ the Cartesian product $\prod_{ i 
%= 1}^n \dom(\U)$.
Capital first letters of the alphabet $A, B, C$ represent attributes, while capital last 
letters such as $X, Y, Z$ denote sets of attributes. 
We write $XY$ as a shortcut for $X \cup Y$, and $A_1A_2A_3$ for $\{A_1, A_2, A_3\}$.
A \emph{tuple} $t$ over $\U$ is a mapping $t \colon \U \to \dom(\U)$ such that $t[A] 
\in \dom(A)$ for each $A \in \U$, where $t[A]$ is the projection of $t$ on $A$.
Similarly, $t[X]$, for some $X \subseteq \U$, denotes the projection of $t$ on $X$.
A \emph{relation} $r$ is a collection of tuples over $\U$.
A \emph{functional dependency} (FD) over $\U$ is a pair $(X, A)$ usually written $X \imp 
A$, where $X \subseteq \U$ and $A \in \U$. 
%If $\U$ is clear from the context, we simply write $X \imp A$. 
Without loss of generality, we assume that $X \imp A$ is not trivial, i.e., $A 
\not\in X$.
Let $r$ be a relation over $\U$. 
We say that $r$ \emph{satisfies} the functional dependency $X \imp A$ if for each pair 
$t_1, t_2$ of tuples in $r$, $t_1[X] = t_2[X]$ implies $t_1[A] = t_2[A]$.
This is written $r \models X \imp A$.
If $\F$ is a set of functional dependencies, we write $r \models \F$ to denote 
that $r \models X \imp A$ for each $X \imp A \in \F$.

\paragraph{Lattices, implications}
Let $S$ be a set. 
A \emph{partial order} or \emph{poset} is a pair $(S, \leq)$, where $\leq$ is a binary 
relation that is reflexive, transitive and anti-symmetric. 
Let $x, y \in S$. 
We say that $x$ and $y$ are \emph{comparable} if $x \leq y$ or $y \leq x$. 
Otherwise, they are \emph{incomparable}. 
A \emph{lower bound} $\ell$ of $x$ and $y$ is an element of $S$ satisfying $\ell \leq 
x$ 
and $\ell \leq y$. 
Moreover, if for any lower bound $\ell'$ of $x$ and $y$ we have $\ell' \leq \ell$, $\ell$ 
is 
called the \emph{meet} of $x$ and $y$. 
We denote it by $x \mt y$. 
Dually, an \emph{upper bound} of $x$ and $y$ is an element $u \in S$ such that $x \leq u$ 
and 
$y \leq u$. 
If $u$ is minimal with respect to all upper bounds of $x$ and $y$, it is the \emph{join} 
of $x$ and $y$. 
We write it $x \jn y$.
For $X \subseteq S$, the meet $\Mt X$ and join $\Jn X$ of $X$ are defined accordingly.
A poset in which every pair of elements has both a meet and a join is a 
\emph{lattice}.
A lattice will be noted $\Lt(S, \leq, \jn, \mt)$ or simply $\Lt$ when precision is
not necessary.
It has a minimum element $0$ called the \emph{bottom}, and a maximum one, the 
\emph{top}, denoted by $1$.
Let $\Lt$ be a lattice, and $x, y, z \in \Lt$ such that $y < x$ (i.e., $y \leq x$ and 
$y \neq x$).
If for any $z \in \Lt$, $y < z \leq x$ implies $x = z$, then we say that $x$ 
\emph{covers} $y$ and we write $y \prec x$.
An element $a \in \Lt$ which covers the bottom $0$ of $\Lt$ is called an \emph{atom}.
We denote by $\cc{A}(\Lt)$ the set of atoms of $\Lt$.
Let $m \in \Lt$.
We say that $m$ is \emph{meet-irreducible} if for any $x, y \in \Lt$, $m = x \mt y$ 
implies $m = x$ or $m = y$.
A \emph{join-irreducible} element $j$ is defined dually with $\jn$.
Meet-irreducible elements (resp. join-irreducible elements) of $\Lt$ are denoted by 
$\cc{M}(\Lt)$ (resp. $\cc{J}(\Lt)$).
Let $\Lt_1, \Lt_2$ be two lattices. 
A map $\varphi \colon \Lt_1 \to \Lt_2$ is \emph{increasing} if for any $x, y \in 
\Lt_1$, $x \leq y$ implies $\varphi(x) \leq \varphi(y)$.
It is a \emph{$\mt$-homomorphism} if for any $x, y \in \Lt_1$, $\varphi(x \mt y) = 
\varphi(x) \mt \varphi(y)$. 
\emph{Decreasing} maps and \emph{$\jn$-homomorphisms} are defined dually. 
The map $\varphi$ is a \emph{homomorphism} if it is both a $\jn$-homomorphism and a 
$\mt$-homomorphism. 
If $\Lt_1, \dots, \Lt_n$ are lattices, the cartesian product $\prod_{i = 1}^{n} 
\Lt_i$ 
with component-wise induced order is the \emph{direct product} of $\Lt_1, \dots, 
\Lt_n$. 
We have $x \leq y$ if and only if $x_i \leq y_i$ for every $1 \leq i \leq n$. 
If $\Lt$ is a lattice, a \emph{$\mt$-sublattice} $\Lt'$ of $\Lt$ is a lattice which 
is a subset of $\Lt$ such that $x, y \in \Lt'$ implies $x \mt y \in \Lt'$.

Now, we introduce lattice implications and closure operators, based on definitions and 
results of \cite{day1992lattice}.
Let $\Lt$ be a lattice and $\cl \colon \Lt \to \Lt$.
The map $\cl$ is a \emph{closure operator} if for any $x, y \in \Lt$, $x \leq \cl(x)$ 
(extensive), $x \leq y \implies \cl(x) \leq \cl(y)$ (increasing) and $\cl(\cl(x)) = 
\cl(x)$ (idempotent).
Then, $\cl(x)$, is the \emph{closure} of $x$ with respect to $\cl$.
The set $\Lt_{\cl} = \cl(\Lt) = \{x \in \Lt \mid \cl(x) = x \}$ of fixed points (or 
closed elements) of $\cl$ is a $\mt$-sublattice of $\Lt$ which contains the top of 
$\Lt$.
Therefore, we also obtain $\cl(x) = \Mt \{y \in \Lt_{\cl} \mid x \leq y \}$.
The set of all such $\mt$-sublattices of $\Lt$ is denoted $\Sb(\Lt)$.
When $\Lt$ equals $2^{\U}$ for some attribute set $\U$, and $\cl$ is a closure 
operator on $\Lt$, $\Lt_{\cl}$ is a \emph{closure system} 
\cite{gratzer2011lattice} (in this case $\mt$ is the set intersection $\cap$).
A \emph{(lattice) implication} is a pair $x \imp y$ such that $x, y \in \Lt$.
We say that an element $z \in \Lt$ satisfies the implication $x \imp y$ if $x \leq z 
\implies y \leq z$, and we write $z \models x \imp y$.
Let $\IS$ be a set of implications over $\Lt$.
We write $z \models \IS$ if $z \models x \imp y$ for each $x \imp y \in \IS$.
We also say that $z$ is a \emph{model} of $\IS$.
Moreover, $\IS$ induces a closure operator $\cl^{\IS}$ on $\Lt$ such that $x \in \Lt$ is
closed if and only if $x \models \IS$.
Formally, $\Lt_{\cl^{\IS}} = \{ x \in \Lt \mid x \models \IS\}$.
Dually, any closure operator $\cl$ on $\Lt$, and hence any $\Lt' \in \Sb(\Lt)$, can 
be represented by implications \cite{day1992lattice}.
Let $\Lt' \in \Sb(\Lt)$ and $x \imp y$ be an implication on $\Lt$.
We obtain $\Lt'$ \emph{satisfies} $x \imp y$, written $\Lt' \models x \imp y$ if 
for 
any $z \in \Lt'$, $z \models x \imp y$.
The notion $\Lt' \models \IS$ follows.
We have the following folklore property: $\Lt' \models x \imp y$ if and only if $y 
\leq \cl(x)$.

\begin{example} \label{ex:ex-lattice}
Let us first consider the lattice $\Lt$, to the left of Figure \ref{fig:ex-lattice}.
The atoms are $a, b$ and $c$. 
%They are also the atoms of $\Lt$.
%Similarly, $d, e$ and $f$ are the meet-irreducible elements of $\Lt$.
Now consider the set $\IS = \{e \imp d, b \imp d\}$ of lattice implications.
The corresponding $\mt$-sublattice $\Lt_{\cl^{\IS}}$ of $\Lt$ is represented to the right 
of Figure \ref{fig:ex-lattice}.
For instance, $f$ is not closed with respect to $\IS$ because it does not satisfy the 
implication $b \imp d$: $b \leq f$ but $d \nleq f$.
The lattice $\Lt_{\cl^{\IS}}$ is also associated to a closure operator $\cl_{\IS}$ 
mapping an element of $\Lt$ to an element $\Lt_{\cl^{\IS}}$.
For example, the closure $\cl_{\IS}(b)$ of $b$ is $d$: it is the unique minimum element 
$x$ of $\Lt_{\cl^{\IS}}$ satisfying $b \leq x$.

\begin{figure}[ht!]
\centering 
\includegraphics[scale=1.0, page=1]{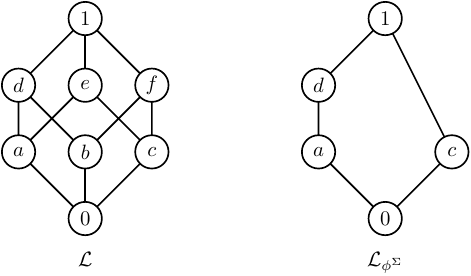}%
\caption{The two lattices of Example \ref{ex:ex-lattice}}
\label{fig:ex-lattice}
\end{figure}
\end{example}

\section{A lattice-based framework to handle comparabilities}
\label{sec:framework}

In this section, we formally define our framework.
Then, we show the relationships between the classical relational model and our framework.

\subsection{Defining the framework}

Let $\U$ be a relation scheme.
For every attribute $A \in \U$, the first step is to define abstract values whenever two 
elements of $\dom(A)$ are compared. 
The \emph{abstract domain} of $A$, denoted by $\DAV{A}$, is a finite set of 
\emph{abstract values} disjoint from the domain $\dom(A)$, representing comparabilities 
between every pair of values in $\dom(A)$.
We assume that a partial order $\leq_A$, or simply $\leq$ when no confusion 
arises, exists over $\DAV{A}$ such that the poset $ \Lt_A = (\DAV{A}, \leq_A)$ is a 
lattice.
We say that $\Lt_A$ is the \emph{abstract lattice} of $A$.
We put $0_A$ and $1_A$ as the bottom and top elements (resp.) of $\Lt_A$.
The \emph{abstract domain} $\DAV{\U}$ of $\U$ is the direct product of 
abstract domains, i.e., $\DAV{\U} = \prod_{A \in \U} \DAV{A}$.
Two attributes can have the same abstract domain.

Now, we define the explicit mapping of pairs of domain values to abstract values through 
a comparability function. 
Let $A$ be an attribute in $\U$.
A \emph{comparability function} for $A$ is a surjective map $f_A$ from $\dom(A) \times 
\dom(A)$ to $\DAV{A}$ which is commutative, i.e., such that $f_A(x_A, y_A) = f_A(y_A, x_A)$ for all $x_A, y_A \in \dom(A)$.

\begin{remark}
When using the framework, one can further require reflexivity on the comparability functions, i.e. $f(x_A, x_A) = 1_A$, in order to get a semantic of comparability closer to equality.
Even more, it could be possible to make the functions reflexive on all values but \ctt{null} where some freedom is allowed.
Indeed, in practice the meaning of the \ctt{null} value in the data should be explained by domain experts, along with recommendations on how to deal with it.
Moreover, since the \ctt{null} value indicates a missing value, relaxing reflexivity of comparability functions on \ctt{null} allows to consider absent values as possibly 
different.
Interestingly, while the results we present in the body of this paper are unchanged by reflexivity, we show in \ref{sec:AFDs} that reflexivity is a key property to ensure completeness of (extended) Armstrong axioms. 
\end{remark}

Our approach ensures that the data remains unchanged, and bypasses usual ways of dealing 
with incomplete data such as completion and repairing \cite{amendola2018explainable, 
bertossi2011database, libkin2016sql}.

Let $t_1, t_2$ be two tuples over $\U$.
The comparability function $f_{\U}$ for $\U$ is the combination of comparability 
functions on attributes of $\U$.
More precisely, for every two tuples $t_1, t_2$ over $\U$, we have:
\[ f_{\U}(t_1, t_2) = \langle f_{A_1}(t_1[A_1], t_2[A_1]), \ldots, f_{A_n}(t_1[A_n], 
t_2[A_n]) \rangle \]
Note that $f_{\U}(t_1, t_2)$ is a generalization of agree sets between two tuples 
\cite{beeri1984structure}.
For any relation $r$ over $\U$, we put $f_{\U}(r) = \{ f_{\U}(t_1, t_2) \mid t_1, t_2 
\in r\}$.
Let us denote by $\Lt_{\U} = (\DAV{\U}, \leq_{\U})$ the direct product of abstract
lattices.
An element of $\DAV{\U}$ is an \emph{abstract tuple}.
It is a vector $\langle a_1, \dots, a_n \rangle$ where $a_i \in \DAV{A_i}$ for every $1 
\leq i \leq n$.
We can use $x, y \in \Lt_{\U}$ as a shortcut for $x, y \in \DAV{\U}$.
Also, we will make a slight abuse of the $[\;]$ notation: if $x \in \Lt_{\U}$ we denote 
by $x[A]$ the projection of the vector $x$ on $A$.
%In the sequel, we may use the shortcut $x \in \Lt_{\U}$ to denote $x \in 
%\DAV{\U}$.
We are now ready to introduce attribute and scheme contexts:

\begin{definition}[attribute context, scheme context] \label{def:context}
Let $\U$ be a relation scheme, and let $A \in \U$.
An \emph{attribute context} for $A$ is a triple $\C_A = (A, f_A, \Lt_A)$ where $\Lt_A = 
(\DAV{A}, \leq_A)$ is an abstract lattice and $f_A \colon \dom(A) \times \dom(A) 
\to \DAV{A}$ a comparability function.
A \emph{scheme context} for $\C_{\U}$ is a union of attribute contexts for each 
$A \in \U$, i.e., $\C_{\U} = \{\C_A \mid A \in \U\}$.
\end{definition}

\textbf{Complexity assumption:} for all the computational results of this paper, 
the contexts are given as part of the input.
For a given attribute context $\C_A = (A, f_A, \Lt_A)$, we assume that $f_A$ can be 
computed in polynomial time in the size of its input.
The lattice $\Lt_A$ is given in such a way that the operations $\leq$, $\jn$, $\mt$, can 
be conducted in polynomial time in the size of $\Lt_A$.

From a scheme context $\C_{\U}$, we have both its associated comparability function 
$f_{\U}$ and its abstract domain $\DAV{\U}$, along with the lattice $\Lt_{\U}$.
If no confusion is possible, the subscript $\U$ will be omitted, i.e., we will use 
$\leq, \mt, \jn$ instead of $\leq_{\U}, \mt_{\U}, \jn_{\U}$, respectively.
In the sequel, a relation $r$ will be defined over a scheme context $\C_{\U}$, 
instead of $\U$. 

\begin{example}
In the running example, the abstract domain of $A$ is 
\[ 
    \DAV{A} = \{\textit{equal}, \textit{distributed}, \textit{similar}, \textit{abnormal}\},
\] 
and the corresponding abstract lattice is $\Lt_A$, illustrated in 
Figure \ref{fig:abstract-lattices}.
The attribute context of $A$ is $(A, f_A, \Lt_A)$. 
The scheme context of $\U$ is then $\C_{\U} = \{(A, f_A, \Lt_A), (B, f_B, \Lt_B), \allowbreak
(C, f_C, \Lt_C) \}$.
For instance, the abstract tuple $\langle s, d_B, c\rangle$ results from the comparison 
of $t_3$, $t_4$ in the relation $\textrm{Patients}$: $ f_{\U}(t_3, t_4) = \langle 
f_A(3.2, 3.5), f_B(M, F), \allowbreak f_C(72, 73)\rangle = \langle s, d_B, c \rangle$.
In Table \ref{tab:comparabilities}, we give (a part of) the family 
$f_{\U}(\textrm{Patients})$ of abstract tuples associated with the relation 
$\textrm{Patients}$.
\begin{table}[ht!]
\centering 
\includegraphics[scale=1.0, page=1]{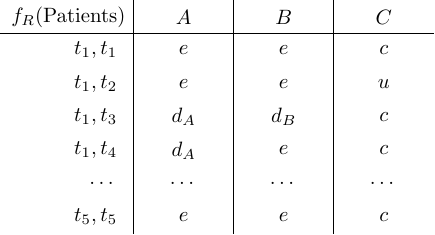}%
\caption{The abstract tuples associated to the relation $\textrm{Patients}$}
\label{tab:comparabilities}
\end{table}
\end{example}

Next, we define (attribute) interpretations as particular mappings from an abstract 
lattice to $\{0, 1\}$.
Their aim is to decide for each two elements in the attribute domain whether they 
can be considered equal.

\begin{definition}[interpretation] \label{def:attr-interp}
Let $\U$ be a relation scheme, $A \in \U$ and $\C_A = (A, f_A, \Lt_A)$ an 
attribute context for $A$. 
An \emph{attribute (context) interpretation} for $A$ is a map $h_A \colon \Lt_A 
\to \{0, 1\}$ satisfying the following properties:
\begin{enumerate}
\item $x \leq y$ in $\Lt_A$ implies $h_A(x) \leq h_A(y)$, i.e., $h_A$ is increasing,	

\item $h_A(1_A) \neq h_A(0_A)$.
\end{enumerate}
\end{definition}

In particular, an interpretation must be surjective, since $h_A(1_A) = 1$ and $h_A(0_A) = 
0$.
Practically, it guarantees that two values in $\dom(A)$ can always be interpreted as 
equal or different.
Moreover, the interpretation has to be increasing. 
Intuitively, if an abstract value $x_A$ of $\Lt_A$ is interpreted as $1$ (i.e., equality) 
by $h_A$, any value $y_A \geq_A x_A$ must be set to $1$ since it is closer to 
\textit{``true''} equality than $x_A$.

\begin{definition}[scheme interpretation] \label{def:sch-interp}
Let $\U = \{A_1, \dots, A_n\}$ be a relation scheme, and $\C_{\U}$ a scheme context.
Let $h_{A_i}$ be an attribute interpretation for $A_i$, for all $1 \leq i \leq n$.
The map $g \colon \Lt_{\U} \to \{0,1\}^n$ defined by $(g(x))[A_i] = h_{A_i}(x[A_i])$, for 
each $x \in \Lt_{\U}$ and each $1 \leq i \leq n$, is a \emph{scheme (context) 
interpretation}.
In the sequel, $g_{\mid A_i}$ will denote the restriction of $g$ to the attribute 
$A_i \in \U$, i.e., $g_{\mid A_i} = h_{A_i}$.
\end{definition}

A scheme interpretation $g$ maps each possible abstract tuple to a binary vector in 
$\{0, 1\}^n$.
This binary vector can be seen as the characteristic vector of some subset of $\U$. 
In the sequel, we will use $g(\cdot)$ interchangeably to denote a set or its 
characteristic vector. 

\begin{example}
We continue our running example.
In Figure \ref{fig:interpretations}, we give three possible scheme interpretations.
Above the dotted lines (white nodes), the abstract values are interpreted as $1$, and $0$ 
otherwise (grey nodes).
For instance, the abstract tuple $\langle e, d_B, u \rangle$ is interpreted as $\langle 1, 
0, 0 \rangle$ with $g_1$ and $g_2$.
It is interpreted as $\langle 1, 0, 1 \rangle$ with $g_3$.
In Table \ref{tab:table-interpretations}, we give (some of) the interpretations of the 
comparison of the tuples of $\textrm{Patients}$.
For example, $g_1(f_{\U}(t_1, t_2)) = g_1(\langle e, e, u \rangle) = \langle 1, 1, 0 
\rangle$ and $g_3(f_{\U}(t_1, t_2)) = \langle 1, 1, 1 \rangle$.
Thus, $t_1$ and $t_2$ are considered equal under the interpretation $g_3$ but different 
for $g_1$ and $g_2$.
\begin{table}[ht!]
\centering 
\includegraphics[scale=1.0, page=1]{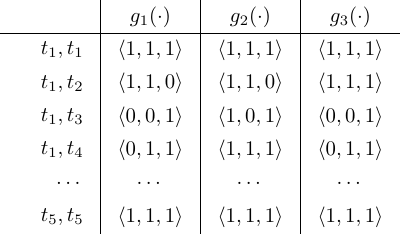}%
\caption{Comparing tuples of $\textrm{Patients}$ with the interpretations of 
Figure \ref{fig:interpretations}}
\label{tab:table-interpretations}
\end{table}
\end{example}

The following statement means that a scheme interpretation $g$ is uniquely defined by 
its components.

\begin{proposition} \label{prop:morphisms}
Let $g$ be a scheme interpretation.
The following properties hold:
\begin{enumerate}
	\item $g$ increasing if and only if $h_{A_i}$ is increasing for any $1 \leq i 
	\leq n$,
	\item $g$ is a $\mt$-homomorphism (resp.~$\jn$-homomorphism) if 
	and only if	for any $1 
	\leq i \leq n$, $h_{A_i}$ is a $\mt$-homomorphism (resp.~$\jn$-homomorphism) .
\end{enumerate}
\end{proposition}

\subsection{Applications to the classical relational model}

In this part, we show how to express the classical relational model in our framework.
We also study the scenario of SQL's 3-valued logic including \ctt{null} values.

\paragraph{Relational model (without nulls)} Given a relation scheme $\U = \{A_1, 
\dots, A_n\}$, we associate to each $A_i \in \U$ the attribute context $(A_i, f_{A_i}, 
\Lt_{A_i})$ where $\Lt_{A_i} = (\{0, 1\}, \leq)$ and $f_{A_i}$ is the strict equality, 
i.e., $f_{A_i}(x, x) = 1$ and $0$ otherwise.
Thus, the abstract tuples associated with a relation are binary vectors over $\{0, 
1\}^n$.
In other words, they are the agree sets of the input relation \cite{beeri1984structure}.
The context $\C_{\U} = \{(A_i, f_{A_i}, \Lt_{A_i}) \mid A_i \in \U\}$ has a unique 
interpretation $g$ such that $g_{\mid A_i}(x[A_i]) = 1$ if $x[A_i] = 1$ and $0$ 
otherwise, for each $A_i \in \U$.
%Observe that $g$ is a strong reality.
The composition  $g(f_{\U})$ precisely depicts the equality between two tuples.
Moreover, this context suggests that in practice, a domain expert can leave classical 
equality as a default comparabability and apply dedicated comparability functions only on 
some relevant attributes.
In this particular setup, possible and certain FDs (see Section \ref{sec:fd}) coincide, 
and they are precisely the valid classical FDs.
Moreover, abstract functional dependencies (see Section \ref{sec:fd}) are the expression 
of classical FDs in terms of characteristic vectors.
As a consequence, an abstract FD is valid if and only if the corresponding FD is true in 
the relation.
It follows that abstract FDs can be seen as a generalization of classical FDs. 
As such, the hardness of problems on abstract FDs, such as identifying a cover of 
abstract FDs for a relation \cite{maier1983theory}, inherits the complexity of their 
counterpart in terms of classical FDs.
The reader can refer to \ref{sec:AFDs} for more details on 
computational aspects of abstract FDs.

\paragraph{SQL model with nulls} 
A similar approach can be applied to model SQL's three-valued logic with \ctt{nulls}.
Namely, if $\U = \{A_1, \dots, A_n\}$ is a relation scheme, we assign to each attribute 
$A_i$ of $\U$ the following attribute context $(A_i, f_{A_i}, \Lt_{A_i})$:
\begin{itemize}
\item $\Lt_{A_i} = (\{ 0, u, 1\}, \leq)$  and $\leq$ is defined by $0 \leq u \leq 1$ ($u$ 
stands for \textit{unknown}),
\item $f_{A_i}(x, x) = 1$ if $x \neq \ctt{null}$, $f_{A_i}(x, y) = u$ if $x$ or $y$ is 
\ctt{null} and $0$ otherwise.
\end{itemize}
Usually, functional dependencies are not defined in the presence of \ctt{null} values.
With the help of interpretations, our framework offers an easy way to tell whether
two tuples are equal in the presence of \ctt{null}.
Even more, it proposes two types of functional dependencies: (i) abstract FDs, 
using the \textit{unknown} value in abstract tuples, and (ii) FDs based on 
interpretations of the context (see Section \ref{sec:fd} for formal a definitions of 
abstract FDs). 

\section{Application to functional dependencies}
\label{sec:fd}

In this section, we study functional dependencies in light of our framework.
First, we show that each relation is associated with an \emph{abstract lattice}, which 
can be represented by lattice implications called \emph{abstract functional 
dependencies}.

Second, we characterize those interpretations that guarantee that any abstract lattice is 
turned into a closure system, thus paving the way for functional dependencies 
\cite{demetrovics1992functional}.
These interpretations are called \emph{(strong) realities}.
Using realities, we establish several results connecting AFDs and classical FDs.
Complementary results regarding abstract FDs can be found in 
\ref{sec:AFDs}.

In the third part of this section, we use realities to introduce \emph{possible} and 
\emph{certain} functional dependencies.
A FD is (strongly) possible if there exists a (strong) reality for which it holds, and 
(strongly) certain if it holds for every (strong) reality.
%Thus, our definition of possible/certain differ from \cite{}.
In particular, we show that the decision as to whether a given FD is possible/certain can 
be conducted in polynomial time.
We prove however that deciding strong possibility is \csf{NP}-complete.

\subsection{Abstract functional dependencies, realities}
\label{subsec:realities}

Let $r$ be a relation over a scheme context $\C_{\U}$.
As a reminder, $\Lt_{\U}$ is the product of abstract lattices $\prod_{i = 1}^{n} 
\Lt_{A_i}$.
The\emph{ abstract lattice} associated with $r$ is 
\[ \Lt_r = \left( \left\{\Mt T \mid T \subseteq f_{\U}(r) \right\}, \leq_{\U} 
\right) \]
As usual, $\Mt \emptyset$ equals the top element of 
$\Lt_{\U}$ \cite{davey2002introduction, day1992lattice}.
Observe that $\Lt_r$ is a subset of $\Lt_{\U}$. 
It is in fact a $\mt$-sublattice of $\Lt_{\U}$.
%We show in the next proposition that it is in fact a $\mt$-sublattice of $\Lt_{\U}$.

\begin{proposition} \label{prop:meet-sublattice}
Let $r$ be a relation over a scheme context $\C_{\U}$.
Then, $\Lt_r \in \Sb(\Lt_{\U})$.
\end{proposition}

\begin{example}
We continue our running example.
The lattice $\Lt_{\U}$ is given in Figure \ref{fig:product}.
We illustrate the abstract lattice associated to the relation $\textrm{Patients}$ in 
Figure \ref{fig:patients-lattice}.
For instance, the abstract tuple $\langle s, d_B, u \rangle$ is obtained by taking $\langle 
s, d_B, c \rangle \mt \langle e, e, u \rangle = f_{\U}(t_3, t_4) \mt f_{\U}(t_1, t_2)$.

\begin{figure}[ht!]
	\centering 
	\includegraphics[scale=0.75, page=1]{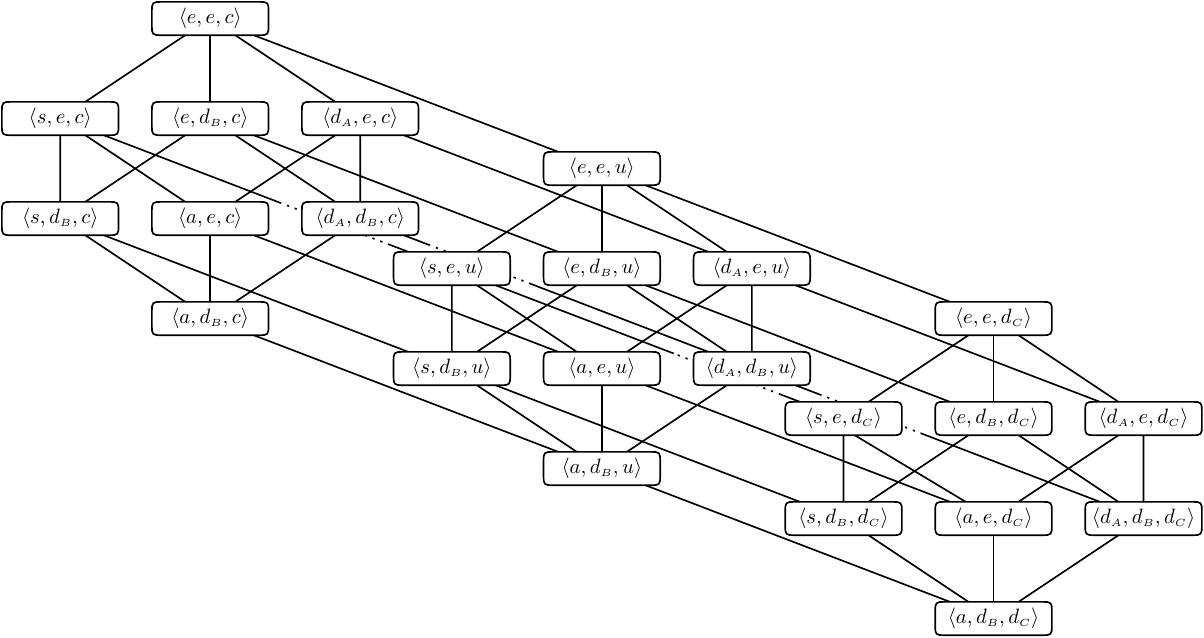}%
	\caption{The product $\Lt_{\U}$ of abstract lattices of the context 
	$\C_{\U}$}
	\label{fig:product}
\end{figure}

\begin{figure}[ht!]
	\centering 
	\includegraphics[scale=0.8, page=2]{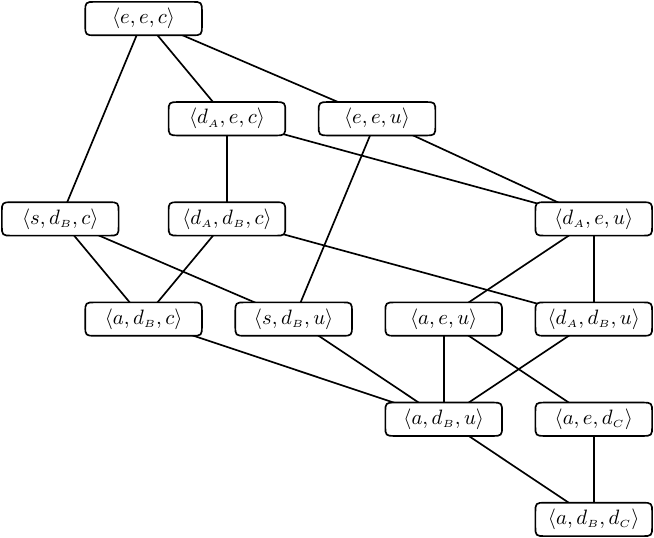}%
	\caption{The abstract lattice associated with the relation 
	$\textrm{Patients}$.}
	\label{fig:patients-lattice}
\end{figure}

\end{example}

%According to Proposition \ref{prop:meet-sublattice}, $\Lt_r$ is a $\mt$-sublattice of 
%$\Lt_{\U}$.
Since $\Lt_r$ is a $\mt$-sublattice of $\Lt_{\U}$, it corresponds to a closure 
operator $\cl$.
%Thus, $\Lt_r$ corresponds to a closure operator, called $\cl$.
Following \cite{day1992lattice}, $\Lt_r$ can be represented by a set $\IS$ of 
lattice implications, the so-called \emph{abstract functional dependencies} in our 
framework.
%On the other hand, $\Lt_r$ can be represented by a set of lattice implications, the 
%so-called \emph{abstract functional dependencies} in our framework.

\begin{definition}[Syntax] \label{def:AFD-syntax}
Let $\U$ be a relation scheme, $\C_{\U}$ a scheme context, and $\DAV{\U}$ its 
abstract domain.
An \emph{abstract functional dependency} over $\C_{\U}$ is an expression of the form 
$x \imp y$ where $x, y \in \DAV{\U}$.
\end{definition}

\begin{definition}[Satisfaction] \label{def:AFD-semantic}
Let $r$ be a relation over $\C_{\U}$ and $x, y \in \DAV{\U}$.
We say that $r$ satisfies $x \imp y$ with respect to $\C_{\U}$, denoted by $r 
\models_{\C_{\U}} x \imp y$, if for all $t_1, t_2 \in r$, $x \leq f_{\U}(t_1, t_2)$ 
entails $y \leq f_{\U}(t_1, t_2)$.
\end{definition}

When no confusion is possible, we will drop the element $\C_{\U}$ of notations regarding 
functional dependencies. Namely, we will denote by $x \imp y$ an abstract functional 
dependency over $\C_{\U}$ and $r \models x \imp y$ when $r$ models $x \imp y$ with 
respect to $\C_{\U}$.
We first establish a natural relationship between the satisfaction of an abstract 
FD in a relation and its validity in the corresponding abstract lattice.

%In \cite{day1992lattice}, the author shows that abstract functional dependencies enjoy 
%extended Armstrong axioms.
%There is a strong relationship between the satisfaction of an abstract FD for a relation 
%and the validity of a lattice implication.

\begin{proposition} \label{prop:sat-abstract-fd}
Let $r$ be a relation over $\C_{\U}$, and $x \imp y$ an abstract functional dependency.
Then $r \models x \imp y$ if and only if $\Lt_r \models x \imp y$.
\end{proposition}

%\begin{proof}
%	The if part is clear as $f_{\U}(r) \subseteq \Lt_r$.
%	Let us assume that $r \models_{\C_{\U}} x \imp y$ and let $z \in \Lt_r$ such that $x 
%	\leq z$.
%	If $z \in f_{\U}(r)$, then $y \leq z$ as $r \models_{\C_{\U}} x \imp y$.
%	Assume $z \in \Lt_r \setminus f_{\U}(r)$.
%	By definition of $\Lt_r$, there exists a subset $F_z \subseteq f_{\U}(r)$ such that 
%	$z = \Mt F_z$.
%	Let $m \in F_z$.
%	As $z \leq m$, we have that $x \leq m$ and since $r \models_{\C_{\U}} x \imp y$, it 
%	must be that $y \leq m$.
%	As this holds for any $m \in F_z$, we have $y \leq \Mt F_z = z$, thus concluding the 
%	proof.
%\end{proof}

\begin{remark}
Since $f_{\U}(r)$ is a generating set of $\Lt_{r}$, and for each $x \in \Lt_{r}$ we 
have $\cl(x) = \Mt \{y \in \Lt_r \mid x \leq y\} = \Mt \{y \in f_{\U}(r) \mid x \leq y\} 
$, it is possible to compute $\cl(x)$ in 
polynomial time in the size of $r$ and $\C_{\U}$ (i.e., if the abstract lattice of each 
attribute is given).
Since an abstract FD $x \imp y$ holds in $\Lt_{r}$ if and only if $y \leq \cl(x)$, 
testing the validity of an abstract FD can also be conducted in polynomial time.
\end{remark}

Abstract functional dependencies depict dependencies between the attributes of a 
relation in the absence of interpretations. 
More precisely, the meaning of an abstract FD $x \imp y$ valid in a relation $r$ is 
\textit{``the similarity of two tuples is greater than $y$ if it is greater than $x$''}.

\begin{example}
We continue our running example.
Let us first consider the abstract FD $\langle e, d_B, d_C \rangle \imp \langle s, e, u 
\rangle$.
To determine whether it holds in $\textrm{Patients}$, we show that $\langle s, e, u 
\rangle \leq \cl(\langle e, d_B, d_C \rangle)$ where $\cl(\langle e, d_B, d_C \rangle)$ is the 
closure of $\langle e, d_B, d_C \rangle$ in the associated abstract lattice (see Figure 
\ref{fig:patients-lattice}).
We have $\cl(\langle e, d_B, d_C \rangle) = \Mt \{x \in f_{\U}(\textrm{Patients}) \mid 
\langle e, d_B, d_C \rangle \leq x\} = \langle e, e, u \rangle$.
Then, $\langle s, e, u \rangle \leq \langle e, e, u \rangle$ and $\langle e, d_B, d_C \rangle 
\imp \langle s, e, u \rangle$ is a valid abstract FD.
Dually, consider the abstract FD $\langle a, e, c \rangle \imp \allowbreak \langle e, d_B, d_C \rangle$.
There exists $x$ in the lattice such that $\langle a ,e ,c \rangle \leq x$ but $\langle 
e, d_B, c \rangle \nleq x$ (take $x = \langle d_A, e, c \rangle$).
Therefore, $\textrm{Patients} \models \langle a, e, c \rangle \imp \langle e, d_B, c 
\rangle$ does not hold.
\end{example}

Now, we examine the connection between abstract functional dependencies and functional 
dependencies.
Let $r$ be a relation over a scheme context $\C_{\U}$.
Its associated abstract lattice is $\Lt_{r}$.
An interpretation $g$ over $\C_{\U}$ maps $\Lt_{r}$ to the set system $g(\Lt_r) = \{g(x) \mid x \in \Lt_r \}$ 
over $\U$.
We illustrate this interpretation in the next example.

\begin{example} \label{ex:abstract-closure}
We consider the scheme context of our running example.
Let $r$ be the relation presented in Table \ref{tab:abstract-closure}.
In Figure \ref{fig:abstract-closure}, we give the abstract lattice $\Lt_r$ and its 
interpretation $g_1(\Lt_r)$ using $g_1$.
Observe that $g_1(\Lt_r)$ is a closure system.
Hence, it can be represented by a set of functional dependencies.
In particular, it satisfies $C \imp A$ and $A \imp B$.

\begin{table}
\centering 
\includegraphics[scale=1.0, page=1]{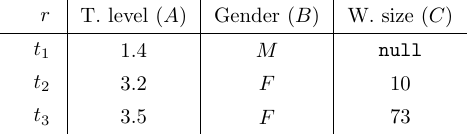}%
\caption{The relation $r$}
\label{tab:abstract-closure}
\end{table}

\begin{figure}[ht!]
\centering 
\includegraphics[scale=0.9, page=2]{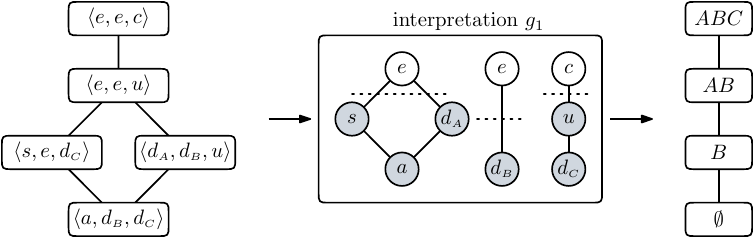}%
\caption{$\Lt_r$ and its interpretation $g_1(\Lt_r)$ using $g_1$}
\label{fig:abstract-closure}
\end{figure}

\end{example}

Thus, the interpretation of an abstract lattice may be a closure system, 
which can be represented by functional dependencies.
However, not all interpretations enjoy this property, as shown by the following example.

\begin{example} \label{ex:abstract-closure-fail}
We continue the previous example.
Instead of $g_1$, we consider the interpretation $g_1'$ depicted in Figure 
\ref{fig:abstract-closure-fail}.
The resulting set system $g_1'(\Lt_r)$ is not closed under intersection.
Henceforth, it is not a closure system, and it cannot be represented by functional 
dependencies.

\begin{figure}[ht!]
\centering 
\includegraphics[scale=0.9, page=3]{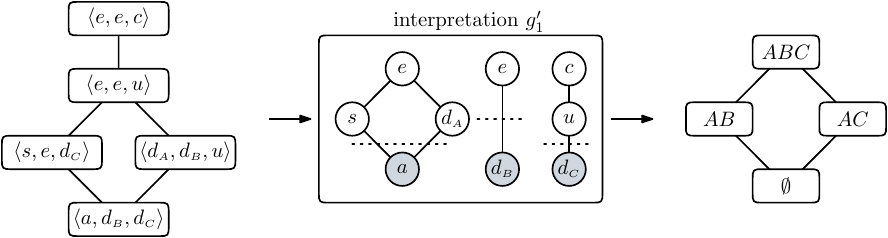}%
\caption{$\Lt_r$ and its interpretation $g_1'(\Lt_r)$ using $g_1'$}
\label{fig:abstract-closure-fail}
\end{figure}
 
\end{example}

Thus, in view of functional dependencies, the following question arises: which 
property should a scheme interpretation have in order to guarantee that the 
interpretation of any abstract lattice is a closure system?
This question is central since functional dependencies must satisfy Armstrong axioms, 
and their model must form a closure system \cite{demetrovics1992functional}.
First, we answer the case where $\U$ has few attributes.

\begin{proposition} \label{prop:realities-small}
Let $\C_{\U}$ be a context scheme with $\U = \{A_1, \dots, A_n\}$.
If $n \leq 2$, then for any scheme interpretation $g$ and any $\Lt \in \Sb(\Lt_{\U})$, 
$g(\Lt)$ is a closure system over $\U$.
\end{proposition}

\begin{proof}
If $n = 2$, then for any $\Lt \in \Sb(\Lt_{\U})$, $g(\Lt)$ is a non-empty subset 
of $\{0, 1\}^2$ containing $\langle 1, 1 \rangle$ by definition of $\Sb(\Lt_{\U})$.
The only subset of $\{0, 1\}^2$ possessing $\langle 1, 1 \rangle$ which is not closed 
with respect to $\mt$ is $\{\langle 0, 1 \rangle, \langle 1, 0 \rangle, \langle1, 
1 \rangle \}$.
However, if $g(\Lt)$ contains $\langle 0, 1 \rangle$ and $\langle 1, 0 \rangle$, 
there exists $x, y \in \Lt$ such that $g(x) = \langle 0, 1 \rangle$ and $g(y) = \langle 
1, 0 \rangle$.
As $\Lt$ is a $\mt$-sublattice of $\Lt_{\U}$, and $g$ is increasing, it must be 
that $x \mt y \in \Lt$, $g(x \mt y) \leq g(x)$ and $g(x \mt y) \leq g(y)$ so that 
necessarily $g(x \mt y) = \langle 0, 0 \rangle$.
Therefore, $g(\Lt)$ is indeed a closure system, concluding the proof.
\end{proof}

In the more general case where the number of attributes is unbounded, this 
property may not hold, as shown by the previous example.
In the next theorem, we characterize such interpretations using lattice homomorphisms.

\begin{theorem} \label{3.2:thm:reality}
Let $\C_{\U}$ be a context scheme with $\U = \{A_1, \dots, A_n\}$, $n \geq 3$, 
and let 
$g$ be a scheme interpretation.
Then $g$ is a $\mt$-homomorphism if and only if for any $\Lt \in \Sb(\Lt_{\U})$, $g(\Lt)$ 
is a closure system over $\U$.
\end{theorem}

\begin{proof}
We begin with the only if part. 
Let $g$ be a $\mt$-homomorphism and $\Lt \in \Sb(\Lt_{\U})$.
Consider $x, y \in \Lt$ such that $g(x)$ and $g(y)$ are incomparable.
Then, since $\Lt \in \Sb(\Lt_{\U})$ and $g$ is a $\mt$-homomorphism, it follows 
that 
$x 
\mt y \in \Lt$ and $g(x \mt y) = g(x) \mt g(y) \in g(\Lt)$.
Furthermore, the top element $\langle 1_{1}, \dots, 1_n \rangle$ belongs to $\Lt$ 
by 
assumption and $g(\langle 1_{1}, \dots, 1_n \rangle) = \U \in g(\Lt)$ by 
definition 
of 
$g$.
Hence $g(\Lt)$ is a closure system over $\U$.

For the if part, we use contrapositive.
Suppose $g$ is not a $\mt$-homomorphism.
We construct $\Lt \in \Sb(\Lt_{\U})$ such that $g(\Lt)$ is not a closure system.
%The lattice $\Lt$ will be built according to the previous example. 
Recall by Proposition \ref{prop:morphisms} that if $g$ is not a $\mt$-homomorphism, 
then 
it must be that at least one of the $h_A$, $A \in \U$ is not a $\mt$-homomorphism 
either. 
Without loss of generality, assume that $A = A_1$.
Since $h_{A_1}$ is increasing but not $\mt$-homomorphic, there must be $x_1, y_1 
\in 
\Lt_{A_1}$ 
such $h_{A_1}(x_1) = h_{A_1}(y_1) = 1$ and $h_{A_1}(x_1 \mt y_1) = 0$.
Recall that we assumed $n \geq 3$.
Let $\Lt$ be the four element set:
\begin{equation*}
	\begin{split}
		\{ 
		\langle x_1, 1_2, 0_3, 1_4, \dots, 1_n \rangle,
		\langle y_1, 0_2, 1_3, 1_4, \dots, 1_n \rangle,
		\langle x_1 \mt y_1, 0_2, 0_3, 1_4, \dots, 1_n \rangle, \\
		\langle 1_1, 1_2, 1_3, 1_4, \dots, 1_n \rangle
		\}	
	\end{split}
\end{equation*}
Note that $\Lt$ is a $\mt$-sublattice of $\Lt_{\U}$ which contains its 
top element.
Since $g$ is increasing and we constrained $h_{A_i}(0_{A_i}) = 0$ and 
$h_{A_i}(1_{A_i}) = 1$, we obtain:
\[g(\Lt) = \{ 
\langle 1, 1, 0, 1, \dots, 1 \rangle,
\langle 1, 0, 1, 1, \dots, 1 \rangle,
\langle 0, 0, 0, 1, \dots, 1 \rangle,
\langle 1, 1, 1, 1, \dots, 1 \rangle 
\} \]
which is not a closure system when interpreted as attribute sets, thus concluding 
the 
proof.
\end{proof}

For each attribute, the $\mt$-homomorphism property defines a minimum element in 
the abstract lattice which is interpreted as $1$. 
Still, $\mt$-homomorphisms only partially capture the intuitive behavior that an 
interpretation could have.
Indeed, while the meet of two abstract values set to $1$ equals $1$ by 
$\mt$-preservation of $g$, the join of two abstract values interpreted as $0$ could be 
$1$.
Hence, we are led to consider two types of interpretations.
First, $\mt$-homomorphisms, which guarantee the semantic for functional dependencies 
thanks to Theorem \ref{3.2:thm:reality}.
Second, homomorphisms, a strengthening of $\mt$-homomorphisms.
These types of interpretations are called realities and strong realities, respectively.

\begin{definition}[reality, strong reality] \label{def:realities}
Let $\C_{\U}$ be a scheme context.
A scheme interpretation $g$ is a \emph{reality} if it is a $\mt$-homomorphism over 
$\Lt_{\U}$.
It is a \emph{strong reality} if it is a homomorphism.
\end{definition}

We denote by $\cc{R}$ the set of realities of $\Lt_{\U}$.
The set of strong realities is $\cc{R}_s$.
For a given reality $g$, let $x_g$ be the abstract tuple of $ 
\Lt_{\U}$ satisfying $x_g[A] = \csf{min}_{\leq}\{x_A \in \Lt_A \mid g_{\mid A}(x_A) = 
1\}$ for every attribute $A$.
As $g$ is an increasing $\mt$-homomorphism, such an $x_g$ is uniquely defined.
Moreover, $x_g[A] \neq 0_A$ for every $A \in \U$.

\begin{remark}
Let $r$ be a relation over a scheme context $\C_{\U}$, and let $g$ be a reality.
As we previously discussed, $\Lt_r$ induces a closure operator from $\Lt_{\U}$ to $\Lt_r$.
Since $g(\Lt_r)$ is a closure system by Theorem \ref{3.2:thm:reality}, it also induces a 
closure operator from the powerset of $\U$ to $g(\Lt_r)$.
In the remaining of this subsection, we write $\cl$ as the closure operator of $\Lt_r$ 
and $\cl_g$ as the closure operator of $g(\Lt_r)$ to avoid confusion.
\end{remark}

\begin{example}
We illustrate in Figure \ref{fig:realities} the six possible realities of our running 
example, along with the interpretations of the abstract lattice associated with the 
relation $\textrm{Patients}$ (see Figure \ref{fig:patients-lattice}).
Among these, $g_2$, $g_3$, $g_5$ and $g_6$ are strong realities.
For example, we have $x_{g_1} = \langle e, e, c \rangle$, $x_{g_2} = \langle d_A, e, c 
\rangle$, and  $x_{g_3} = \langle s, e, u \rangle$.
\begin{figure}[ht!]
\centering
\includegraphics[scale=0.8, page=1]{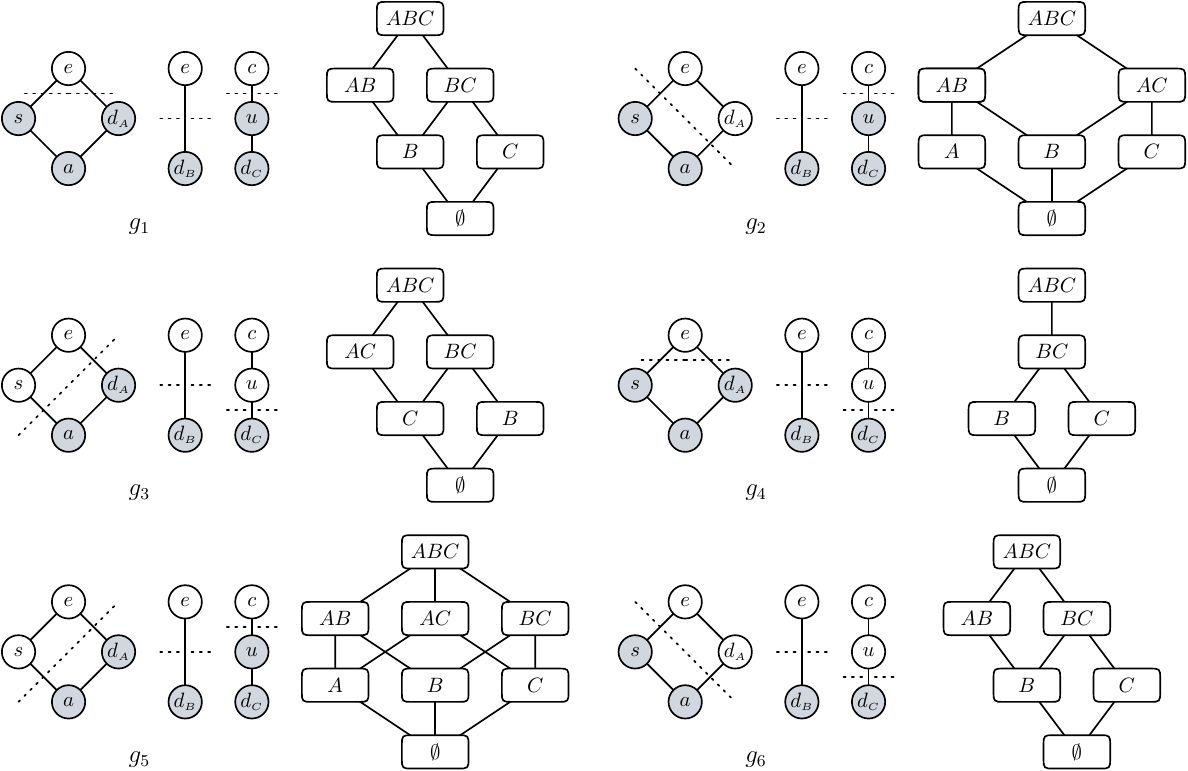}%
\caption{Realities and interpretations of the abstract lattice 
associated with the relation $\textrm{Patients}$}
\label{fig:realities}
\end{figure}
\end{example}

Now we give some properties of strong realities.
First, unlike normal realities, there may be cases where $\cc{R}_s = 
\emptyset$, as shown by the following example.

\begin{example}
Let $\U = \{A, B, C\}$ and $\Lt_A, \Lt_B, \Lt_C$ be the abstract lattices of Figure 
\ref{fig:ex-no-strong}. 
Due to $\Lt_A$ no strong reality is possible: for any attribute interpretation 
$h_A$, 
each element $a, b, c$ of $\Lt_A$ must be assigned a binary value so that 
at least two elements will have the same value, say $a, b$.
However, the homomorphism property implies that $h_A(a \jn b) = h_A(a) \jn h_A(b) 
= h_A(a) \mt h_A(b) = h_A(a \mt b)$ so that either $h_A(0) = 1$ or $h_A(1) = 0$, a 
contradiction 
with the definition of an interpretation.
\begin{figure}[ht!]
	\centering
	\includegraphics[scale=1.0]{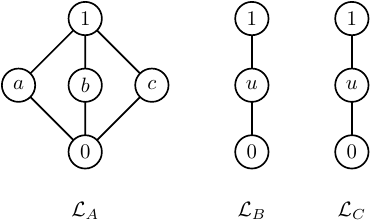}
	\caption{abstract lattices without strong realities}
	\label{fig:ex-no-strong}
\end{figure}
\end{example}

Let us assume that $\cc{R}_s$ is non-empty.
We show that strong realities are related to prime/co-prime decomposition pairs 
in lattices \cite{markowsky1992primes}.
Let $\Lt$ be a lattice, and $x \in \Lt$. 
We say that $x$ is \emph{prime} if for any $y, z \in \Lt$, $y \mt z \leq x$ implies 
either $y \leq x$ or $z \leq x$.
The term \emph{co-prime} is defined dually with $\jn$.
The bottom element of $\Lt$ cannot be co-prime, and the top element 
cannot be prime.
It is known from \cite[Theorem 6]{markowsky1992primes}, that $\Lt$ has a prime element 
$p$ 
if and only if it has a co-prime element $c$.
Actually, if $p$ is prime, then it has a unique associated co-prime element $c = 
\csf{min}_{\leq}\{x \in \Lt \mid x \nleq p\}$.
Dually, if $c$ is co-prime, it has a unique associated prime element $p = 
\csf{max}_{\leq}\{x \in \Lt \mid x \ngeq c\}$.
This correspondence between primes and co-primes is bijective.
Furthermore, a pair $(c, p)$ splits $\Lt$ into two disjoint parts  $\{x \in \Lt \mid x 
\geq c \} $ and $\{ x \in \Lt \mid x \leq p\}$ such that $\{x \in \Lt \mid c \leq x 
\}  
\cup \{ x \in \Lt \mid x \leq p\} = \Lt$.
For any $A \in \U$, if $\Lt_A$ has a pair $c, p$ of coprime/prime elements, we shall 
denote it $(c, p)_A$.
Furthermore, the set of co-prime (resp.~ prime) elements of $\Lt_A$ is denoted by 
$\coPm(\Lt_A)$ (resp.~ $\Pm(\Lt_A)$).
Now, we establish the relationship between strong realities and co-prime/prime 
pairs.

\begin{proposition} \label{prop:coprime-prime}
	Let $\C_{\U}$ be a scheme context, and $g$ a reality. 
	Then, $g$ is a strong reality if and only if for any $A \in \U$, there are 
	co-prime/prime pairs $(c, p)_A$ such that $c = \csf{min}_{\leq}\{x_A \in \Lt_A 
	\mid 
	g_{\mid A}(x_A) = 1\}$ and $p = \csf{max}_{\leq}\{x_A \in \Lt_A \mid g_{\mid 
	A}(x_A) =  
	0\}$.
\end{proposition}

\begin{proof}
We begin with the only if part.
Let $g \in \cc{R}_s$ and consider $g_{\mid A}$ for some $A \in \U$.
We show that $g_{\mid A}$ induces a co-prime/prime pair $(c, p)_A$.
We put $c = \Mt g_{\mid A}^{-1}(1)$ and $p = \Jn g_{\mid A}^{-1}(0)$.
Note that as $g(0_A) = 0$ and $g(1_A) = 1$ by definition, $g_{\mid A}^{-1}(1)$ 
and 
$g_{\mid A}^{-1}(0)$ are non-empty.
Furthermore, since $g$ is increasing and homomorphic, both $p$ and $c$ are 
well-defined and satisfy $g_{\mid A}(p) = 0$ and $g_{\mid A}(c) = 1$.
Let us prove that $p$ is prime. 
Let $x_A, y_A \in \Lt_A$, such that $x_A \mt y_A \leq p$.
As $x_A \mt y_A \leq p$, it must be that $g_{\mid A}(x_A \mt y_A) = 0$ and either 
$g_{\mid A}(x_A) = 0$ or $g_{\mid A}(y_A) = 0$ since otherwise we would break the 
$\mt$-homomorphic property of $g$.
By construction of $p$ then, it follows that either $x_A \leq p$ or $y_A \leq p$.
A similar reasoning can be applied to $c$.
We show that $\csf{min}_{\leq}\{x_A \in \Lt_A \mid x_A \nleq p\}$ is a singleton 
element 
and that it equals $c$.
Let $y_A, z_A \in \{x_A \in \Lt_A \mid x_A \nleq p\}$.
By construction of $p$, and since $g$ is an homomorphism, it must be that $g_{\mid 
A}(y_A) 
= g_{\mid_A}(z_A) = 1$ so that $y_A \mt z_A \in \{x_A \in \Lt_A \mid x_A \nleq 
p\}$.
Therefore, this set has a unique minimum element.
By definition of $c$, we also have that $c \leq y_A \mt z_A$ and $c = 
\csf{min}_{\leq}\{x_A \in \Lt_A \mid x_A \nleq p\}$ holds.
Similarly we obtain $p = \csf{max}_{\leq}\{x_A \in \Lt_A \mid x \ngeq c\}$.
Thus $(c, p)_A$ is indeed a co-prime/prime pair of $\Lt_A$.
As these arguments can be applied for any $A \in \U$, the only if part of the 
statement follows.

Now we move to the if part.
Let $A \in \U$ and $(c, p)_A$ be a co-prime/prime pair of $\Lt_A$.
Let us put $h_A$ as follows:
\[
h_A(x_A) = \begin{cases}
	1 & \text{ if } x_A \geq c \\
	0 & \text{ if } x_A \leq p \\
\end{cases}
\]
We have to check that such a definition covers the whole lattice $\Lt_A$ in a 
satisfying way.
First, Let $x_A \in \Lt_A$ such that $x_A \nleq p$.
As $p$ is prime and $(c, p)_A$ is a co-prime/prime pair, we have that $c = 
\csf{min}_{\leq}\{y_A \in \Lt_A \mid y_A \nleq p\}$ so that $x_A \geq c$.
Hence, any $x_A \in \Lt_A$ is uniquely mapped to $0$ or $1$ and $h_A$ is an 
increasing 
mapping.
As for the homomorphism property, we show the $\jn$-homomorphic side of $h_A$.
As $h_A$ is increasing, the case where $h_A(x_A) = h_A(y_A) = 1$ is already 
covered.
However, if $h_A(x_A) = h_A(y_A) = 0$, then we have $x_A \leq p$ and $y_A \leq p$ 
so that $x_A \jn y_A \leq p$ and hence $h_A(x_A \jn y_A) = 0$ by definition of $h_A$.
The $\mt$-homomorphic side of $h_A$ is proved analogously.
Now, if we consider $g = \langle h_{A_1}, \dots, h_{A_n}  \rangle$ we obtain a strong 
reality in virtue of Proposition \ref{prop:morphisms} and Definition 
\ref{def:realities}.
\end{proof}

Thus, any strong reality is the result of the choice of a co-prime/prime pair in each 
abstract lattice.
More formally, for any $A \in \U$ and any $c \in \coPm(\Lt_A)$, let us define 
$h_A^c(x_A) = 1$ if $x_A \geq c$ and $0$ if $x_A \leq p$ where $p$ is the prime 
element 
associated with $c$ in $\Lt_A$.
Then, we have:
\[
\cc{R}_s = \prod_{A \in \U} \{h_A^c \mid c \in \coPm(\Lt_A) \}
\]

% 
%======================================================================================
% %
% ---- Abstract FDs 
%-------------------------------------------------------------------- %
% 
%======================================================================================
% %

With (strong) realities, we can now study the relationship between abstract FDs and 
classical FDs.
As we have seen, any relation is associated with an abstract lattice from which abstract 
FDs can be derived.
Furthermore, any (strong) reality gives a semantic for classical FDs when the 
abstract lattice is interpreted.
First, we settle the notion of satisfaction of an abstract FD through a reality.

\begin{definition}[Satisfaction] \label{def:sat-fd-reality}
Let $\C_{\U}$ be a scheme context, $r$ a relation over $\C_{\U}$, $g$ a reality 
and $x \imp y$ an abstract FD.
The functional dependency $g(x) \imp g(y)$ is satisfied in $r$, denoted by $r 
\models_{\C_{\U}, g} g(x) \imp g(y)$ (or simply $r \models_{g} g(x) \imp g(y)$) 
if for all $t_1, t_2 \in r$, $g(x) \subseteq g(f_{\U}(t_1, t_2))$ implies $g(y) 
\subseteq g(f_{\U}(t_1, t_2))$.
\end{definition}

Note that in the interpretation $g(x) \imp g(y)$ of an abstract FD $x 
\imp y$, $g(y)$ may coincide with a set $Y$ instead of an attribute $A$.
Without loss of generality, a FD of the form $X \imp Y$ should be read as the 
equivalent set of FDs $\{ X \imp A \mid A \in Y \}$.

We first give a proposition similar to Proposition \ref{prop:sat-abstract-fd}.
This allows us to work directly on the lattice $\Lt_r$ induced by $r$ without 
loss of generality with respect to FDs.

\begin{proposition} \label{prop:sat-fd-reality-1}
Let $r$ be a relation over $\C_{\U}$, $g$ a reality, and $x \imp y$ an abstract 
FD.
We have $r \models_g g(x) \imp g(y)$ if and only if for any $z \in \Lt_r$, $g(x) 
\subseteq g(z)$ implies $g(y) \subseteq g(z)$, denoted $\Lt_r \models_g g(x) \imp g(y)$.
\end{proposition}

In what follows, we will need the next definition.

\begin{definition}\label{def:projection-g}
Let $g$ be a reality.
The \emph{projection} $\pi_g(x)$ of $x \in \Lt_{\U}$ on $g$ is given as 
follows: for any $A \in \U$, $\pi_g(x)[A] = x_g[A]$ if $x[A] \geq x_g[A]$, and 
$0_A$ otherwise.
\end{definition}

We prove in Theorem \ref{thm:projection-g} that the interpretation of a valid abstract 
FD of $\Lt_r$ is true for $g$ if its projection on $g$ is still a valid abstract FD for 
$\Lt_r$.
We need the following technical lemma first.

\begin{lemma} \label{lem:projection-g}
Let $r$ be a relation over a scheme context $\C_{\U}$.
Let $x \in \Lt_{\U}$, and $g$ be a reality. 
Then $g(\cl(\pi_g(x))) = \cl_g(g(\pi_g(x))) = \cl_g(g(x))$.
\end{lemma}

\begin{theorem} \label{thm:projection-g}
Let $r$ be a relation over a scheme context $\C_{\U}$.
Let $g$ be a reality, and $x \imp y$ an abstract FD.
Then $\Lt_r \models_g g(x) \allowbreak \imp g(y)$ if and only if $\Lt_r \models 
\pi_g(x) \imp \pi_g(y)$.
\end{theorem}

\begin{proof}
Let $g$ be a reality, and let $x \imp y$ be an abstract FD.
We prove the only if part.
Let us assume that $\Lt_r \models_g g(x) \imp g(y)$.
As $g(x) = g(\pi_g(x))$ and $g(y) = g(\pi_g(y))$, we also have that $\Lt_r 
\models_g 
g(\pi_g(x)) \imp g(\pi_g(y))$.
We show that $\Lt_r \models \pi_g(x) \imp \pi_g(y)$.
Since $\Lt_r \models_g g(\pi_g(x)) \imp g(\pi_g(y))$, we have that $g(\pi_g(y)) 
\subseteq 
\cl_g(g(\pi_g(x))) = g(\cl(\pi_g(x)))$ by Lemma \ref{lem:projection-g}.
We prove that $\pi_g(y) \leq \cl(\pi_g(x))$.
Let $A \in \U$.
If $A \notin g(\pi_g(y))$, then we have $\pi_g(y)[A] = 0_A$ by construction of 
$\pi_g$ 
and hence $\pi_g(y)[A] \leq_A \cl(\pi_g(x))[A]$.
If $A \in g(\pi_g(y)) \subseteq g(\cl(\pi_g(x)))$, then $\pi_g(y)[A] \leq_A 
\cl(\pi_g(x))[A]$ holds again by construction of $\pi_g$.
Therefore,  $\pi_g(y) \leq \cl(\pi_g(x)) $ and $\Lt_r \models \pi_g(x) \imp 
\pi_g(y)$ follows.

We move to the if part.
Let us assume that $\Lt_r \models \pi_g(x) \imp \pi_g(y)$.
We have that $\pi_g(y) \leq \cl(\pi_g(x))$.
Then, $g(\pi_g(y)) \subseteq g(\cl(\pi_g(x)))$ as $g$ is increasing, and 
$g(\cl(\pi_g(x))) = \cl_g(g(\pi_g(x))) = \cl_g(g(x))$ by Lemma \ref{lem:projection-g}.
As $g(\pi_g(y)) = g(y)$, we have that $g(y) \subseteq \cl_g(g(x))$ and $\Lt_r 
\models_g g(x) \imp g(y)$ follows.
\end{proof}

Moreover, if an abstract FD holds, then it reflects a potential dependency on $r$, 
meaning that there should be a reality which turns it into a valid FD.
This is the aim of the next proposition.

\begin{proposition} \label{prop:valid-afd}
Let $\C_{\U}$ be a schema context, $r$ be a relation over $\C_{\U}$, and $x \imp 
y$ be an abstract FD.
If $r \models x \imp y$, then there exists a reality $g$ such that $r \models_g g(x) \imp 
g(y)$.
\end{proposition}

\begin{proof}
Assume that $r \models x \imp y$.
By Proposition \ref{prop:sat-fd-reality-1}, this is equivalent to $\Lt_r \models x \imp 
y$ where $\Lt_r$ is the 
abstract lattice associated with $r$.
Let $g$ be a reality satisfying $\pi_g(x) = x$.
Note that such a reality always exists as it only requires to put $x_g[A] = x[A]$ 		
for every $A \in \U$ such that $x[A] \neq 0_A$.
Then, by Lemma \ref{lem:projection-g}, $g(\cl(x)) = \cl_g(g(x)$.
Since $\Lt_r \models x \imp y$, we have that $y \leq \cl(x)$ and therefore $g(y) 
\subseteq g(\cl(x)) = \cl_g(g(x))$ so that $\Lt_r \models_g g(x) \imp g(y)$.
We deduce that $r \models g(x) \imp g(y)$ as expected.
\end{proof}

To conclude this subsection, we illustrate Proposition \ref{prop:valid-afd} on our 
running example.

\begin{example}
We consider our running example.
Let $x \imp y$ be the abstract FD where $x = \langle e, d_B, u \rangle$ and $y =  
\langle s, e, d_C \rangle$.
It is a valid abstract FD of the lattice associated to $\textrm{Patients}$ (and hence of 
$\textrm{Patients}$).
Thus, according to Proposition \ref{prop:valid-afd}, there must exist a reality $g$ for 
which the 
FD $g(x) \imp g(y)$ is valid in $r$. 
This is the case, for instance, of the reality $g_1$ depicted to the left of Figure 
\ref{fig:interpretations} (see also Figure \ref{fig:realities}). 
More precisely, we have $g_1(x) = A$, $g_1(y) = B$, and $r \models_{g_1} A \imp B$.
\end{example}

\subsection{Possible and certain FDs}
\label{sec:possible}

In this part, we are interested in functional dependencies from another 
point of view.
We are given a relation $r$ over a scheme context $\C_{\U}$, and a functional 
dependency 
$X \imp A$.
Let us remind that, in our terms, a functional dependency $X \imp A$ means a pair $(X, 
A)$ where $X \subseteq \U$ and $A \in \U$. 
Since any reality $g$ maps $\Lt_r$ to a closure system, it is natural to wonder about the 
functional dependencies holding in $g(\Lt_r)$. 
Beforehand, we formally define the meaning of a FD holding through a reality. 
This will mainly be a reformulation of Definition \ref{def:sat-fd-reality} with attribute 
sets 
instead of abstract tuples.

\begin{definition} \label{def:fd-glr}
Let $r$ be a relation over $\C_{\U}$, $g$ a reality, $X \subseteq \U$, and $A \in 
\U$.
Then, $X \imp A$ is satisfied in $r$ with respect to $\C_{\U}$ and $g$, denoted $r 
\models_{g} X \imp A$, if for every $t_1, t_2 \in r$,  $X \subseteq g(f_{\U}(t_1, 
t_2))$ implies $A \in g(f_{\U}(t_1, t_2))$.
\end{definition}

Proposition \ref{prop:sat-abstract-fd} can also be reformulated in this way, thus 
introducing the notation $\Lt_r \models_g X \imp A$ if for any $x \in \Lt_r$, $X 
\subseteq g(x)$ implies $A \in g(x)$.

\begin{proposition} \label{prop:sat-fd-reality-2}
	Let $r$ be a relation over $\C_{\U}$, $g$ a reality, $X \subseteq \U$, and $A \in 
	\U$.
	We have $r \models_g X \imp A$ if and only if $\Lt_r \models_g X \imp A$.
\end{proposition}

\begin{proof}
	Let $x, a \in \Lt_{\U}$ such that $g(x) = X$ and $g(a) = A$.
	Note that $x, a$ must exist as $g(\Lt_{\U}) = 2^{\U}$.
	The result follows from \ref{prop:sat-abstract-fd} applied to $x, a$.
\end{proof}

There are many different realities for a given $\C_{\U}$. 
Therefore, the same FD $X \imp A$ may or may not hold depending on the reality we 
choose. 
However, there are also some FDs that will hold for every reality such as the trivial 
case $X \imp A$, $A \in X$.
Thus, inspired by certain query answering, we define certain functional dependencies 
as those that are true for any reality.
Similarly, a functional dependency is possible if it holds for at least one reality.
They are defined along with decision problems.

\begin{definition}[certain FD, strongly certain FD]
Let $r$ be a relation over $\C_{\U}$, and $X \imp A$ a functional dependency.
We say that $X \imp A$ is \emph{certain} (resp.~ \emph{strongly certain}) if for 
any reality $g \in \cc{R}$ (resp.~ strong reality $g \in \cc{R}_s$), $r \models_g X 
\imp A$ holds.
\end{definition}

\Problem{Certain Functional Dependency (CFD)}
{A scheme context $\C_{\U}$, a relation $r$ over $\C_{\U}$, and a functional 
dependency $X 
	\imp A$.}
{Is $X \imp A$ certain with respect to $r, \C_{\U}$?}

\Problem{Strongly Certain Functional Dependency (SCFD)}
{A scheme context $\C_{\U}$, a relation $r$ over $\C_{\U}$, and a functional 
	dependency $X \imp A$.}
{Is $X \imp A$ strongly certain with respect to $r, \C_{\U}$?}

\begin{definition}[possible FD, strongly possible FD]
Let $r$ be a relation over $\C_{\U}$, and $X \imp A$ a functional dependency.
We say that $X \imp A$ is \emph{possible} (resp.~ \emph{strongly possible}) if 
there exists a reality $g \in \cc{R}$ (resp.~ a strong reality $g \in \cc{R}_s$) such 
that $r \models_g X \imp A$ holds.
\end{definition}

\Problem{Possible Functional Dependency (PFD)}
{A scheme context $\C_{\U}$, a relation $r$ over $\C_{\U}$, and a functional 
dependency $X \imp A$.}
{Is $X \imp A$ possible with respect to $r, \C_{\U}$?}

\Problem{Strongly Possible Functional Dependency (SPFD)}
{A scheme context $\C_{\U}$, a relation $r$ over $\C_{\U}$, and a functional 
dependency $X \imp A$.}
{Is $X \imp A$ strongly possible with respect to $r, \C_{\U}$?}

\begin{remark}
All the complexity statements of this section are based on the complexity 
assumption given after Definition \ref{def:context}.
\end{remark}

Before discussing the complexity of these problems we give an illustration of 
possible and certain functional dependencies.

\begin{example}
We continue the running example.
In Table \ref{tab:fd-holding}, we list the (left-minimal) non-trivial FDs holding in 
$\textrm{Patients}$, 
according to the realities $g_1$ to $g_6$ given in Figure \ref{fig:realities}.
It follows that all of them are possible, since they hold in at least one reality.
On the other hand, no non-trivial FD holds in every reality.
Thus, certain FDs of $\textrm{Patients}$ are of the form $X \imp Y$ where $Y \subseteq X$.
\begin{table}[ht!]
\centering 
\includegraphics[scale=1.0, page=1]{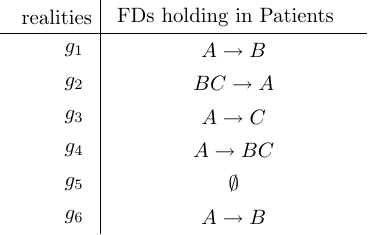}%
\caption{Non-trivial FDs holding in $\textrm{Patients}$, depending on the 
reality}
\label{tab:fd-holding}
\end{table}
\end{example}

\begin{remark}
A functional dependency $X \imp A$ which is valid in the classical sense 
(with respect to equality) needs not be possible in a given scheme context.
Whether or not $X \imp A$ will be possible depends on how the comparability functions are 
given.
In practice, this is not a problem: if the comparability functions provided by 
domain experts do not allow to identify classical functional dependencies, it means that 
usual FDs may not be relevant with respect to their data.
\end{remark}

Now, we show that the number of realities can be exponential. 

\begin{proposition} \label{4:prop:expo}
Let $\U$ be a relation scheme, and $\C_{\U}$ a scheme context. 
Then, the number of (strong) realities can be exponential in the size of $\U$ and 
$\C_{\U}$.
\end{proposition}

\begin{proof}
Let $\U$ be a relation scheme and $\C_{\U} = \{(A, f_A, \Lt_A) \mid A \in \U \}$ 
where $\Lt_A$ is the three-element lattice $0 \leq u \leq 1$ for each $A \in \U$.
Note that the size of $\C_{\U}$ is bounded by the size of $\U$.
Now, for each $\Lt_A$, there are only two possible interpretations $h_A^1$ and 
$h_A^2$:
\begin{enumerate}
	\item $h_A^1(0) = h_A^1(u) = 0$ and $h_A^1(1) = 1$,
	\item $h_A^2(0) = 0$ and $h_A^2(u) = h_A^2(1) = 1$.
\end{enumerate}
Clearly both $h_A^1$ and $h_A^2$ are homomorphisms.
Using Proposition \ref{prop:morphisms} and Theorem \ref{3.2:thm:reality} we 
can describe $\cc{R}$ as follows:
\[ \cc{R} = \left\{\langle h_{A_1}^{i_1}, h_{A_2}^{i_2}, \dots, h_{A_n}^{i_n} 
\rangle 
\mid i_k 
\in \{1, 2 \}, 1 \leq k \leq n \right\} 
\]
Thus, $\cc{R}$ is in bijection with all binary words of size $n$ on the $\{1, 2\}$ 
alphabet. 
Hence $\card{\cc{R}} = 2^n$.
\end{proof}

It follows that the problems of deciding whether $X \imp A$ is possible or certain cannot 
be solved in polynomial time with a brute force algorithm testing all the existing 
realities.

\subsubsection{Certain functional dependencies}

Now, we investigate the problems related to certain functional dependencies.
Note that thanks to Proposition \ref{prop:sat-fd-reality-2} we can use equivalently $r 
\models X \imp A$ and $\Lt_r \models X \imp A$.
To do so, we characterize certainty of a FD using the structure of $\Lt_r$.
As a reminder, $\coPm(\Lt_A)$ (resp.~ $\Pm(\Lt_A)$) denotes the set of co-prime 
(resp.~ prime) elements of the abstract lattice $\Lt_A$, for any $A \in \U$.

\begin{lemma} \label{lem:certain-FD}
	Let $r$ be a relation over $\C_{\U}$, and $X \imp A$ a functional dependency.
	Then:
	\begin{enumerate}
		\item $X \imp A$ is certain if and only if for any $x \in 
		\Lt_r$, either $x[A] = 1_A$ 
		or there exists $B \in X$ such that $x[B] = 0_B$.
		\item $X \imp A$ is strongly certain if and only if for any $x \in \Lt_r$
		either $x[A] \geq_A \Jn \coPm(\Lt_A)$ or there exists $B \in X$ such that 
		$x[B] 
		\ngeq_B c$ for any $c \in \coPm(\Lt_B)$.
	\end{enumerate}
\end{lemma}

Therefore, deciding whether a FD is certain amounts to checking the elements of $\Lt_r$.
Let $X \subseteq \U$.
We define $\psi_X \in \Lt_{\U}$ where $\psi_X[A] = 1_A$ if $A \in X$, and $0_A$ 
otherwise.
In other words, $\psi_X$ is the characteristic vector of $X$ in $\Lt_{\U}$.
We deduce in the next theorem that \csmc{CFD} can be solved in polynomial time.

\begin{theorem} \label{thm:CFD-P}
	\csmc{CFD} can be solved in polynomial time.
\end{theorem}

\begin{proof}
Let $(r, \C_{\U}, X \imp A)$ be an instance of \csmc{CFD}.
From Lemma \ref{lem:certain-FD}, we have to check that for any element $x$ in 
$\Lt_r$, either $x[A] = 1_A$ or there exists $B \in X$ such that $x[B] = 0_B$.
Let us set up $x \in \Lt_{\U}$ as follows:
\begin{itemize}
	\item $x[A] = 1_A$,
	\item $x[B] = 0_B$ for any $B \neq A$ (in particular for any $B \in X$)
\end{itemize}
We consider $\cl(x)$.
Remark that as for any $B \in \U$, $\Lt_B$ is given in the input, computing the $\mt$ 
operation in $\Lt_{\U}$ can be conducted in polynomial time in the size of $\C_{\U}$ by 
running over the elements of $\Lt_B$.

Note that any element $y \in \Lt_r$ which satisfies $y[A] = 1_A$ does not require 
further verification and furthermore, lies above $\cl(x)$ by construction of $x$.
Hence it remains to check that for any $y \ngeq \cl(x)$ in $\Lt_r$, there exists 
some $B \in X$ such that $x[B] = 0_B$.
We show that we only need to test this property for the set $\csf{max}_\leq\{y 
\ngeq \cl(x) \mid y \in \Lt_r\}$.
First, assume that there exists $m \in \csf{max}_\leq\{y \ngeq \cl(x) \mid y \in \Lt_r\}$ 
such that for any $B \in X$, $m[B] \neq 0_B$.
Then, the condition of Lemma \ref{lem:certain-FD} fails with $m$ and hence $X \imp A$ 
is not certain.
Now suppose that for any $m \in \csf{max}_\leq\{y \ngeq \cl(x) \mid y \in \Lt_r\}$, 
there exists some $B \in X$ such that $m[B] = 0_B$.
By definition, for any $y \ngeq \cl(x)$ there must exists some $m \in \csf{max}_\leq\{y 
\ngeq \cl(x) \mid y \in \Lt_r\}$ such that $y \leq m$.
As by assumption we supposed that there exists some $B \in X$ such that $m[B] = 
0_B$ and as $y \leq m$, $y[B] = 0_B$ follows.
Thus, for any $y \ngeq \cl(x)$ in $\Lt_r$, there exists some $B \in X$ such that 
$x[B] = 0_B$ if and only if for any $m \in \csf{max}_\leq\{y \ngeq \cl(x) \mid y \in 
\Lt_r\}$, there exists $B \in X$ such that $m[B] = 0_B$.

It remains to show that the set $\csf{max}_\leq\{y \ngeq \cl(x) \mid y \in \Lt_r\}$ 
can be checked efficiently.
First, observe that $\csf{max}_\leq\{y \ngeq \cl(x) \mid y \in \Lt_r\} \subseteq 
\cc{M}(\Lt_r)$.
Then, by construction of $\Lt_r$, any element $y$ in $\Lt_r \setminus f_{\U}(r)$ 
is obtained by taking the meet of a subset of $f_{\U}(r)$.
Thus, $y \notin \cc{M}(\Lt)$, and $\cc{M}(\Lt_r) \subseteq f_{\U}(r)$.
As the size of $f_{\U}(r)$ is bounded by $\card{r^2}$ and each verification takes 
$O(\card{\U})$ operations, we conclude that checking all the 
tuples in $f_{\U}(r)$ can be done in polynomial time in the size of $(r, 
\C_{\U}, X \imp A)$, thus concluding the proof.
\end{proof}

We prove a similar theorem for strongly certain functional dependencies.
However, in this case, we have to add a pre-processing step to compute co-prime
elements of the attribute lattice $\Lt_A$, for any $A \in \U$.

\begin{theorem} \label{thm:SCFD-P}
\csmc{SCFD} can be solved in polynomial time.
\end{theorem}

\begin{proof}
	Let $(r, \C_{\U}, X \imp A)$ be an instance of \csmc{SCFD}.
	We first need to compute, for any $B \in \U$, the set $\coPm(\Lt_B)$.
	As $\Lt_B$ is given in the input $\C_{\U}$, this can be done in polynomial time 
	in 
	the 
	size of $\C_{\U}$ by greedily checking for each element $x_B$ in $\Lt_B$ whether 
	$\csf{max}_{\leq}\{y_B \in \Lt_B \mid y_B \ngeq x_B \}$ is a singleton set.
	Therefore, we can already check that strong realities exist by Proposition
	\ref{prop:coprime-prime} in polynomial time.
	If there are no strong realities, the answer to \csmc{SCFD} is positive.
	
	If there are strong realities, we use the arguments of Theorem \ref{thm:CFD-P} 
	with $x \in \Lt_{\U}$ as follows:
	\begin{itemize}
		\item $x[A] = \Jn \coPm(\Lt_A)$,
		\item for any other $B \in \U$, $B \neq A$, $x[B] = 0_B$.
	\end{itemize}
	Note that instead of checking whether for any $m \in \csf{max}_\leq\{y \ngeq \cl(x) 
	\mid 
	y 
	\in \Lt_r\}$, there exists some $B \in X$ such that $m[B] = 0_B$, we check that 
	there 
	exists $B \in X$ such that $m[B] \ngeq c$ for any $c \in \coPm(\Lt_B)$.
	This amounts equivalently to assert that $m[B] < \Mt \Pm(\Lt_B)$ by definition of 
	co-prime/prime pairs decomposition.
	Again, as $\Lt_B$ is part of the input, computing $\Pm(\Lt_B)$ and $\Mt 
	\Pm(\Lt_B)$ 
	can be done in polynomial time in the size of $(r, \C_{\U}, X \imp A)$.
\end{proof}

\subsubsection{Possible functional dependencies}

Now, we settle the complexity of the problems \csmc{PFD} and \csmc{SPFD}.
Again, we need an intermediary structural lemma.
It is the dual of Lemma \ref{lem:certain-FD}.

\begin{lemma} \label{lem:possible-FD}
	Let $r$ be a relation over $\C_{\U}$, and $X \imp A$ a functional dependency.
	Then:
	\begin{enumerate}
		\item $X \imp A$ is not possible if and only if there exists $x \in \Lt_r$, 
		such that $x[A] = 0_A$ and $x[B] = 1_B$ for any $B \in X$.
		\item $X \imp A$ is strongly possible if and only if there exists $c_X \in 
		\Lt_{\U}$ such that $c_X[B] = c$, for some $c$ in $\coPm(\Lt_B)$, for every 
		$B 
		\in X$, $c_X[B] = 0_B$ for every $B \notin X$, and 
		such that $\cl(c_X)[A] \geq c$ for some $c \in \coPm(\Lt_A)$.
	\end{enumerate}
\end{lemma}

Hence, for a given FD $X \imp A$, \csmc{PFD} can be solved by checking $\cl(\psi_X)$. 
As a reminder, $\psi_X$ is the characteristic vector of $X$ in $\Lt_{\U}$, that is 
$\psi_X[A] = 1_A$ if $A \in X$, and $0_A$ otherwise.

\begin{theorem} \label{thm:PFD-P}
	\csmc{PFD} can be solved in polynomial time.
\end{theorem}

\begin{proof}
Let $(r, \C_{\U}, X \imp A)$ be an instance of \csmc{PFD}.
According to Lemma \ref{lem:possible-FD}, $X \imp A$ is not possible if and only if 
there is an element $x$ in $\Lt_r$ such that $x[B] = 1_B$ for any $B \in X$ and $x[A] = 
0_A$.
Since $f_{\U}(r)$ generates $\Lt_r$, we have that $\cl(\psi_X) = \bigwedge \{x \in 
f_{\U}(r) 
\mid \psi_X \leq x\}$.
Moreover, the abstract lattices of each attribute context are part of the input so that 
computing $f_{\U}(r)$ and the $\mt$ operation in each lattice can be conducted in 
polynomial time.
Thus, we can compute the closure $\cl(\psi_X)$ of $\psi_X$ (note that $\psi_X[A] = 0_A$) 
in polynomial time in the size of $r$ and $\C_{\U}$.
Then, either $\cl(\psi_X)[A] = 0_A$ in which case $X \imp A$ is not possible, or 
$\cl(\psi_X)[A] \neq 0_A$ and $X \imp A$ is possible.
\end{proof}

However, for strong possibility, one might have to test an exponential number 
of closures since we have to check for any possible combination of co-prime elements.
In fact, we prove in the following theorem that \csmc{SPFD} is intractable (reduction 
from \csmc{3SAT}).

\begin{theorem} \label{thm:SPFD-NP}
	\csmc{SPFD} is \csf{NP}-complete.
\end{theorem}

\begin{proof}
First we show that \csmc{SPFD} belongs to \csf{NP}.
A certificate is a strong reality $g \colon \Lt_{\U} \to \{0,1\}^n$ which can be 
represented by the abstract tuple $x_{g}$ of polynomial size.
Moreover, the satisfaction of a functional dependency $X \imp A$ can be tested in 
polynomial time by computing $g(f_{\U}(t, t'))$ for every pair of tuples $t, t'$ 
in the input relation $r$.
Thus, \csmc{SPFD} belongs to \csf{NP}.

To show \csf{NP}-hardness, we use a reduction from \csmc{3-SAT}.
Let $\varphi = \bigwedge_{j = 1}^m C_j$, $m \in \cb{N}$, be a 3-CNF over a set of 
variables 
$V = \{x_1, \dots, x_n\}$, $n \in \cb{N}$.
We assume without loss of generality that no clause of $\varphi$ contains both $x_i$ and 
$\bar{x_i}$, for each $1 \leq i \leq n$.
We construct a scheme context $\C_{\U}$, a relation $r$ and a functional dependency $X 
\imp A$ such that $\varphi$ is satisfiable if and only if there exists a strong reality 
$g$ 
(in $\C_{\U}$) such that $r \models_g X \imp A$.

First, we introduce an attribute $A_i$ for each $x_i \in V$.
Then, we define a relation scheme $\U = \{A_1, \dots, A_{n + 1}\}$ where $A_{n + 1}$ is 
yet another attribute.
For each $A_i \in \U$, we put $\dom(A_i) = \cb{N}$.
Now we define a scheme context $\C_{\U}$ as follows:
\begin{itemize}
\item for each $1 \leq i \leq n$, we put $\DAV{A_i} = \{0_i, a_i, \bar{a_i}, 
1_i\}$, 
$\Lt_{A_i} = (\DAV{A_i}), \leq_i)$ where $\leq_i$ is defined by $0 \leq_i a_i 
\leq 1_i$ and $0 \leq_i \bar{a_i} \leq_i 1$.
We define the comparability function $f_i$ associated to $A_i$ as follows:
\[ f_i(x,y) = \begin{cases}
	1_i       & \text{ if } x = y \\
	a_i       & \text{ if } \card{x - y} = 1 \\
	\bar{a_i} & \text{ if } \card{x - y} = 2 \\
	0_i       & \text{ otherwise.}
\end{cases} \]
The attribute context of $A_i$ is $(A_i, f_i, \Lt_{A_i})$.

\item for $A_{n + 1}$, we put $\DAV{A_{n + 1}} = \{0_{n+1}, 1_{n+1}\}$, 
$\Lt_{A_i} = (\DAV{A_{n + 1}}), \leq_{n + 1})$ where $\leq_{n + 1}$ is defined 
by $0 \leq_{n + 1} 1$.
The comparability function $f_{n + 1}$ associated to $A_{n + 1}$ reads as follows:
\[f_{n + 1}(x,y) = \begin{cases}
	1_{n+1} & \text{ if } \card{x - y} \neq 1 \\ 
	0_{n+1} & \text{ if } \card{x - y} = 1 \\
\end{cases} \]
The attribute context of $A_{n + 1}$ is $(A_{n + 1}, f_{n + 1}, \Lt_{A_{n + 1}})$.
\end{itemize}
Let $\C_{\U} = \{(A_i, f_i, \cc{L}_{A_i}) \mid 1 \leq i \leq n + 1\}$ be the 
resulting scheme context.
Its comparability function is called $f_{\U}$.
Now, we construct a relation $r$ over $\cc{C}_{\U}$.
To every clause $C_j$ in $\varphi$, $1 \leq j \leq m$ we associate a subrelation 
$r_j$ with two tuples $t$ and $t'$:
\begin{itemize}
\item for each $A_i, 1 \leq j \leq n$:
\begin{itemize}
	\item if $x_i \in C_j$, $t[A_i] = 3j-2$ and $t'[A_i] = 3j$, 
	\item if $\bar{x_i} \in C_j$, $t[A_i] = 3j-2$ and $t'[A_i] = 3j-1$,
	\item $t[A_i] = t'[A_i]=3j$ otherwise.
\end{itemize}
\item $t[A_{n+1}]=3j-1$ and $t'[A_{n+1}]=3j$.
\end{itemize}
Let $r = \bigcup_{1 \leq j \leq m} r_j$.
We consider the functional dependency $X \imp A$ where $X=\{A_1,...,A_n\}$ and $A 
= A_{n+1}$.
Clearly, the size of the reduction is polynomial in the size $\varphi$. 
The strong realities of the abstract context coincide with the Cartesian product $\{a_i, 
\bar{a_i}\}^n$.

Let $g$ be a strong reality and for every $1 \leq i \leq n + 1$, let $h_i$ be the 
projection of $g$ on the attribute $A_i$, \textit{i.e.} $h_i = g_{\mid A_i}$.
First, we show that $r \models_g X \imp A$ if and only if $r_j \models_g X \imp A$ for 
each subrelation $r_j$, $1 \leq j \leq m$.
The only if part is clear, since $r \models_g X \imp A$ and $r_j \subseteq 
r$ entails $r_j \models_g X \imp A$.
For the if part, it is sufficient to show that any pair of tuples lying in different 
subrelations always agree on $A$ with respect to $g$. 
Let $t, t'$ be two tuples of $r$ such that $t$, $t'$ are in distinct subrelations 
$r_j$, $r_k$ where $1 \leq j < k \leq m$ (resp.~).
By construction of $r$, the minimum value of $\card{t[A_{n + 1}] - t'[A_{n + 
1}]}$ is reached when $t[A_{n + 1}] = 3j$ and $t'[A_{n + 1}] = 3k - 1$.
As $k \geq j + 1$, we obtain $3k - 1 \geq 3j + 2$ and hence $\card{t[A_{n+1}] - 
t'[A_{n+1}]} \geq 2$.
Therefore, we have $f_{n + 1}(t[A_{n+1}],t'[A_{n+1}])= 1_{n + 1}$ by definition of 
$f_{n + 1}$ and $h_{n + 1}(f_{n+1}(t[A_{n+1}],t'[A_{n+1}])) = 1$ for each (strong) 
reality $g$.
Thus, for every reality, whenever two tuples $t$ and $t'$ disagree on the 
right-hand side of $X \imp A$, they must belong to the same subrelation.

Now we prove that $\varphi$ is satisfiable if and only if $X \imp A$ is strongly possible.
We begin with the only if part.
Suppose that $\varphi$ is satisfiable and let $\mu: V \to \{0,1\}$ be a valid 
truth assignment of $\varphi$. 
We construct a strong reality $g$ such that $r \models_g X \imp A$. 
For every attribute context $(A_i, f_i, \Lt_{A_i})$, $1 \leq j \leq n$, we 
define the following interpretation $h_i$:
\begin{itemize}
	\item $h_i(0_i) = 0$ and $h_i(1_i) = 1$,
	\item $h_i(a_i)=\mu(x_i)$ and $h_i(\bar{a_i}) = 1 - \mu(x_i)$.
\end{itemize}
For $A_{n+1}$, we define $h_{n+1}(0_{n + 1}) = 0$ and $h_{n+1}(1_{n + 1}) = 1$.
Let $g: \Lt_{\U} \to \{0,1\}^n$ be the scheme interpretation $g: \Lt_{\U} \to \{0,1\}^n$ 
defined by $g(\langle x_1, \ldots, x_{n+1}\rangle) = \langle h_1(x_1), \ldots, 
\allowbreak h_{n+1}(x_{n+1})\rangle$. 
For every $1 \leq j \leq n$, we have $h_i(a_i) \neq h_i(\bar{a_i})$ which 
implies that $h_i$ is an homomorphism.
Therefore, $g$ is a strong reality by Proposition \ref{prop:morphisms} and 
Definition \ref{def:realities}.
We show that $r \models_g X \imp A$.
Using previous discussion, it is sufficient to prove that $r_j \models_g X \imp A$ 
for all $1 \leq j \leq m$ to obtain $r \models_g X \imp A$.
Hence, let $t, t'$ be the two tuples of $r_j$ for some $C_j = (\ell_1 \lor \ell_2 \lor 
\ell_3)$ in $\varphi$.
At least one literal in $C_j$, say $\ell_1$, validates $C_j$ with respect to $\mu$. 
Since $\ell_1 \in \{x_i,\bar{x_i}\}$ for some $A_i \in \U$, we have two cases:
\begin{enumerate}
	\item $\ell_1 = x_i$.
	Then $t[A_i]=3j-2$, $t'[A_i]=3j$ and $f_i(t[A_i], t'[A_i])=\bar{a_i}$. 
	Moreover, $h_i(\bar{a_i}) = 1 - \mu(x_i) = 0$ since $\mu(x_i) = 1$.
	Consequently, $h_i(f_i(t[A_i], t'[A_i])) = 0$ and $r_j \models_g X 
	\imp A$ holds as $A_i \in X$.
	
	\item $\ell_1 = \bar{x_i}$.
	Then $t[A_i]=3j-2$, $t'[A_i]=3j-1$ and $f_i(t[A_i], t'[A_i]) = a_i$. 
	Since $h_i(a_i)=\mu(x_i) = 0$, we deduce $h_i(f_{A_i}(t[A_i], t'[A_i])) = 0$ 
	and $r_j \models_g X \imp A$ follows. 
\end{enumerate}
In other words, the tuples of every subrelation $r_j$ cannot agree on $X$. 
We conclude that $X \imp A$ is always satisfied in the strong reality $g$ and 
that $X \imp A$ is strongly possible.

We move to the if part.
Assume that $X \imp A$ is strongly possible and let $g$ be a strong reality satisfying $r 
\models_g X \imp A$.
Since $g$ is a strong reality and by Proposition \ref{prop:morphisms}, we 
have $h_i(a_i) \neq h_i(\bar{a_i})$.
As a consequence, we can define a truth assignment $\mu \colon V \to \{0, 1\}$ as follows:
\[ 
\mu(x_i) = \begin{cases}
1 & \text{ if } f_i(t[A_i], t[A_i']) = \bar{a_i} \text{ for some } t, t' \in r \text{ 
and } h_i(\bar{a_i}) = 0 \\
0 & \text{ if } f_i(t[A_i], t[A_i']) = a_i \text{ for some } t, t' \in r \text{ 
and } h_i(a_i) = 0 \\
0 & \text{ otherwise. } 
\end{cases}
\]
Note that $f_i(t[A_i], t[A_i']) \in \{a_i, \bar{a_i}\}$ if and only if $t, t'$ are 
the two tuples of the same subrelation $r_j$, for some $1 \leq j \leq m$.
We show that $\mu$ is a satisfying assignment of $\varphi$.
Let $C_j$ be a clause of $\varphi$ and let $t, t'$ be the two tuples of $r_j$.
By assumption, $r_j \models_g X \imp A$.
Moreover, $h_{n + 1}(f_{n + 1}(t[A_{n + 1}], t'[A_{n + 1}])) = 0$ by construction of 
$f_{n + 1}$.
Therefore, there must exist $A_i \in X$ such that $h_i(f_i(t[A_i], t'[A_i])) = 0$.
We have two cases:
\begin{itemize}
\item $f_i(t[A_i], t'[A_i]) = a_i$ in which case $\bar{x_i}$ belongs to $C_j$ by 
construction of $f_{\U}$.
As $h_i(a_i) = 0$, we have $\mu(x_i) = 0$ by definition.
Hence $C_j$ is satisfied by $\mu$.

\item $f_i(t[A_i], t'[A_i]) = \bar{a_i}$ so that $x_i$ belongs to $C_j$.
Since $h_i(\bar{a_i}) = 0$, $\mu(x_i) = 1$ and $C_j$ is satisfied by $\mu$.
\end{itemize}
In any case, $\mu$ is a satisfying truth assignment for $C_j$.
Applying the same reasoning to every clause of $\varphi$, we deduce that $\mu$ is a 
satisfying truth assignment for $\varphi$, concluding the proof.
\end{proof}

%We conclude this section with a summary of the complexity results we established for 
%certain and possible functional dependencies.

We summarize our complexity results in Table \ref{tab:complexity}.

\begin{table}[ht!]
\centering 
\begin{tabular}{c | c | c}
	& Realities & Strong Realities \\ \hline
	Certain FD & $\csf{P}$ (Theorem \ref{thm:CFD-P}) & $\csf{P}$ 
	(Theorem \ref{thm:SCFD-P}) \\
	Possible FD & $\csf{P}$ (Theorem \ref{thm:PFD-P}) & $\csf{NP}$-complete 
	(Theorem \ref{thm:SPFD-NP})
\end{tabular}
\caption{Complexity results for Possible and Certain FD problems with respect 
	to realities and strong realities}
\label{tab:complexity}
\end{table}

To conclude this section, we briefly discuss the complexity of the possibility and 
certainty problems for sets of FDs.
Let $\C_{\U}$ be a scheme context, $r$ be a relation over $\C_{\U}$ and $\F$ be a set of 
FDs also over $\C_{\U}$.
For a given reality $g$, we write $r \models_g \F$ if $r \models_g X \imp A$ for each $X 
\imp A$ in $\F$.
If there exists a (strong) reality $g$ such that $r \models_g \F$, then $\F$ is 
(strongly) possible in $r$.
If $r \models_g \F$ for any (strong) reality $g$, then $\F$ is (strongly) certain in $r$.

For certainty, it is clear that $\F$ is (strongly) certain if and only if each $X \imp 
A$ 
is (strongly) certain itself.
Therefore, using Theorem \ref{thm:CFD-P} and Theorem \ref{thm:SCFD-P}, we directly 
deduce 
that it takes polynomial time to check that $\F$ is (strongly) certain.

We move to possibility.
Deciding whether a single FD $X \imp A$ is strongly possible is already \NP-complete due 
to Theorem \ref{thm:SPFD-NP}.
It readily follows that checking whether $\F$ is strongly possible is also \NP-complete.
The case of simple possibility however seems harder to settle.
To date, it remains an intriguing open question.
we give some hints on its hardness.
Formally, the problem reads as follows:

\Problem{Possible Set of Functional Dependencies (PSFD)}
{A scheme context $\C_{\U}$, a relation $r$ and a set $\F$ of functional dependencies, 
both over $\C_{\U}$.}
{\ctt{yes} if there exists a reality $g$ such that $r \models_g \F$, \ctt{no} otherwise.}

In order to study the complexity of \csmc{PSFD}, one could try to characterize the case 
where a set of FDs is not possible, much as in Lemma \ref{lem:possible-FD} for a single 
FD.
For a single FD $X \imp A$, the fact that $X \imp A$ must be false in each reality allows 
us to identify an abstract tuple which characterizes the possibility of the FD $X \imp A$.
Unfortunately, this approach does not seem to apply to a set $\F$ of FDs, since there is 
in general no FD $X \imp A$ of $\F$ which is false in each reality (see Example 
\ref{ex:possible-sets}). 
This makes a characterization similar to Lemma \ref{lem:possible-FD} harder to obtain.

%In order to study the complexity of \csmc{PSFD}, one could use the approach of Lemma 16
%determining the possibility of a single FD.
%In the proof of Lemma \ref{lem:possible-FD}, we use the fact that a \emph{single} FD $X 
%\imp A$ must be false  in each possible reality to identify a prototypical abstract 
%tuple which is always interpreted as a counter-example to $X \imp A$.
%This approach cannot be straightforwardly applied to \csmc{PSFD} since in each reality, 
%\emph{at least one} FD must be false.
%Thus, it may be that for two different realities $g_1$ and $g_2$, a same abstract tuple 
%will be a counterexample to some FD $X \imp A$ in $g_1$ but not in $g_2$, in particular 
%if $X \imp A$ holds in $g_2$, as illustrated in Example \ref{ex:possible-sets}.
%This makes a characterization similar to Lemma \ref{lem:possible-FD} harder to obtain.

Another strategy is based on the observation that $\F$ cannot be possible if there is 
some $X \imp A$ in $\F$ which is not possible.
However, unlike for certain FDs, the other direction is not true.
More precisely, it may happen that all FDs in $\F$ are possible as singletons, but that 
$\F$ as a whole is not possible.
Again, Example \ref{ex:possible-sets} illustrates this situation.
As a consequence, using a decomposition of $\F$ and solve sub-problems to obtain an 
answer to \csmc{PSFD} on $\F$ seems unpromising.
This observation also suggests that in general, if one obtains a valid reality for 
each FD of $\F$ (or for each part of some partition of $\F$), there is no guarantee that 
these realities can be combined to obtain a valid reality for $\F$.
This is illustrated in Example \ref{ex:possible-sets}.

\begin{example} \label{ex:possible-sets}
Let $\C_{\U} = \{(A, f_A, \Lt_A), (B, f_B, \Lt_B)\}$ be a scheme context where $\dom(A) 
= 
\dom(B) = \cb{N} \cup \{\ctt{null}\}$, $\Lt_A$, $\Lt_B$ are the abstract lattices given 
in 
Figure \ref{fig:possible-1}, and $f_A$, $f_B$ are comparability functions defined as 
follows:

\[ 
f_A(u, v) = f_B(u, v) = \begin{cases}
1 & \text{if } u = v \neq \ctt{null} \\
a & \text{if } u \neq \ctt{null}, v \neq \ctt{null} \text{ and } u \neq v \\
b & \text{if } u = v = \ctt{null} \\
0 & \text{otherwise.}
\end{cases}
\]
\begin{figure}[ht!]
\centering 
\includegraphics[scale=1.0, page=1]{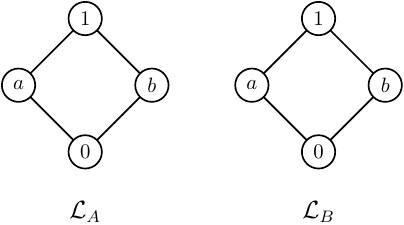}%
\caption{The abstract lattices of $\C_{\U}$}
\label{fig:possible-1}
\end{figure}
Let $r$ be the relation on $\C_{\U}$ given to the left of Figure \ref{fig:possible-2}.
Its abstract lattice $\Lt_r$ is represented to the right.
\begin{figure}[ht!]
\centering 
\includegraphics[scale=1.0, page=2]{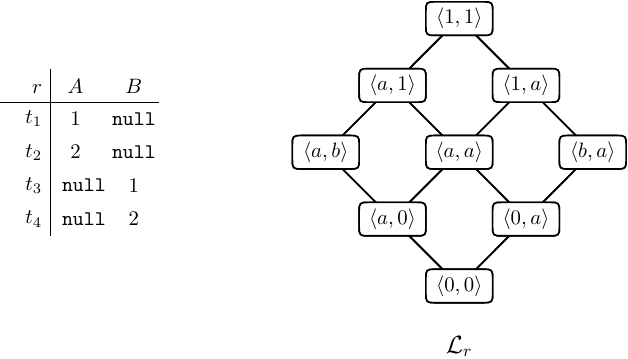}%
\caption{A relation $r$ and its associated abstract lattice $\Lt_r$}
\label{fig:possible-2}
\end{figure}
We consider the set of FDs $\F = \{A \imp B, B \imp A\}$.
For convenience we give in Table \ref{tab:possible-3} the list of all possible realities 
along with the abstract tuples which will be interpreted as counter-examples to $A \imp 
B$ or $B \imp A$.
Explicitly, a set in $g(\Lt_r)$ is a counter-example to a FD $X \imp A$ if it contains 
$X$ but not $A$.

First, remark that both $A \imp B$ and $B \imp A$ are possible.
Indeed, if we set $g = \langle b, a \rangle$ or $g = \langle a, 1 \rangle$, then $r 
\models_g A \imp B$ as there are no counter-examples in the resulting closure system.
Similarly $r \models_g B \imp A$ if $g = \langle a, b \rangle$ or $\langle a, 1 \rangle$.
However, there is no reality in which both FDs hold true together.
Therefore, $\F$ is not possible in $r$.
We can also make the following observations:
\begin{itemize}
\item among all the counter-examples of $A \imp B$, none of them appears as a 
counter-example in all realities where $A \imp B$ does not hold. 
As a consequence, it is not sufficient to exhibit a unique prototypical 
counter-example to each FD (as in Lemma \ref{lem:possible-FD}) to characterize the 
possibility of $\F$.
\item the realities $\langle a, 1 \rangle$ and $\langle 1, a \rangle$ make $B \imp A$ 
and $A \imp B$ valid (respectively).
Moreover both $\langle a, 1 \rangle \mt \langle 1, a \rangle$ and $\langle a, 1 \rangle 
\jn \langle 1, a \rangle$ are realities.
But since $\F$ is not possible, neither of them are valid realities for $\F$.
This suggests that the rather natural process of combining realities with lattice 
operations is not sufficient to determine the possibility of $\F$.
\end{itemize}
\begin{table}[ht!]
\centering 
\includegraphics[scale=1.0, page=3]{Figures/PSFD-hard.pdf}%
\caption{Table of counter-examples to $\F$ according to the different realities}
\label{tab:possible-3}
\end{table}
\end{example}

\section{Related work}
\label{sec:related}

Incomplete information has been extensively studied in the database and artificial 
intelligence communities, see for example 
\cite{abedjan2018data, abiteboul1995foundations, bertossi2011database,  console2019we, 
greco2014certain, imielinski1989incomplete, lenzerini2002data, link2019relational, 
	suciu2011probabilistic}.
Numerous research on data dependencies has been conducted over the last years,  
leading to a plethora of propositions from seminal FDs to more elaborated forms of 
dependencies, among which we quote 
\cite{baixeries2018characterizing, bertossi2013data, caruccio2015relaxed, 
	demetrovics1992functional, goguen1967fuzzy, link2019relational, ng2001extension}.

Many papers have studied the lattice representation of functional dependencies, 
such as \cite{day1992lattice, demetrovics1992functional}, which has been
extended to multivalued dependencies by Balcazar and Baixeries in \cite{baixeries2005new}.
However, they do not consider incomplete information as we do in this paper.
W. Ng \cite{ng2001extension} defined ordered domains in the relational data model, 
i.e., a partial order over every attribute domain is permitted.
The paper studies the consequences on both FDs and relational languages (algebra and 
SQL). 
His work does not consider incomplete information as we do: our partial order is not 
defined on the attribute domain, but on the abstract domain of attributes, and is 
required to form a lattice. 
It offers a new point of view on FDs in presence of incomplete information.

In \cite{baixeries2018characterizing}, order dependencies are based on a transitive 
relation, and approximate dependencies on a symmetric relation, leading to 
approximate-matching dependencies.
Unlike our framework, the comparison of two values in their model is a one-step process 
based on Boolean similarity functions (reflexive and symmetric).
The set with similarities of Bauer and Hajdinjak \cite{hajdinjak2009similarity} 
are close to attribute contexts, but the authors do not consider interpretations, neither 
do they study functional dependencies in this setup.
In \cite{bertossi2013data}, the authors study matching dependencies using Boolean  
similarity functions and matching functions (idempotent, commutative, and associative). 
More recently, the authors of \cite{schirmer2020efficient} study matching 
dependencies where the similarity functions have values in the interval $[0, 1]$.
The aim of matching functions is to chase the relation instance to obtain a clean 
relation. 
However, our way of dealing with incomplete information is completely different.
In \cite{libkin2016sql}, the author studies the semantics of SQL query answering in 
presence of incomplete information. 
He defines a multi-valued logic similar to our contribution, without considering data 
dependencies.
Based on SQL three-valued model, a semantic for possible/certain (or weak/strong) 
FDs have been studied in several works \cite{al2022strongly, kohler2016possible, 
kohler2018sql, levene1998axiomatisation, lien1982equivalence}.
These works either rely on a completion of the data, or implicit interpretations (in our 
sense) of the similarity of two values at least one of which is \ctt{null} (see also the 
very recent work \cite{libkin2020handling}).
In our contribution, we do not modify the input data, and the interpretation of the 
comparison with a \ctt{null} value as true or false ($0$ or $1$) in the SQL model appears 
as a particular case of our construction (see Section \ref{sec:framework}).

To evaluate the truth of a proposition like equality of two entities, it is usual 
to order the truth values in a lattice \cite{arieli1996reasoning}. 
In addition to the well-known true/false lattice, several other semantics that 
might be useful have been studied.
This is the case, for example, of Kleene or \L{ukasiewicz} 3-valued logic 
\cite{bolc2013many}, and Belnap's 4-valued logic, where truth values can be ordered by 
both their degree of truth and their degree of knowledge, or informativeness 
\cite{belnap1977useful}. 
Belnap's logic has been generalized to bi-lattices that are more suitable for 
non-monotonic reasoning \cite{arieli1996reasoning, ginsberg1988multivalued}.
In \cite{console2016approximations}, the authors define many-valued semantics and 
informativeness ordering for query answers in incomplete databases with missing
values represented by marked \ctt{null}. 
Recently, Grahne \cite{grahne2018useful} used the 4-valued logic to capture 
inconsistencies in databases and for query answering in the presence of inconsistencies.

Several works extend the Codd's relational model using fuzzy logic.
Among these, we quote \cite{buckles1982fuzzy, prade1984generalizing, raju1988fuzzy, 
chen1992fuzzy, bhuniya1993lossless, cubero1994new, tyagi2005complete, kiss1990x, 
yahia1999extension, bosc1999database, belohlavek2011codd, belohlavek2018relational}.
We redirect the reader to the survey of Jevskova et al. \cite{jevzkova2017fuzzy} 
for an in-depth comparison of all these approaches.
They replace the classical Boolean logic at the core of the logical foundations of the 
relational model by a fuzzy setup.
More precisely, they replace the true/false Boolean lattice by a more complex set of 
truth values, still ordered in a lattice.
Usually, this lattice is a (continuous) chain lattice on the interval 
$[0, 1]$, except for the framework of Belohlavek and Vychodil which consider every 
lattice that can be residuated, being more general \cite{belohlavek2011codd, 
belohlavek2018relational}.
Then, they adapt the logical operations of conjunction and implication to this new 
set of truth values, thus obtaining a fuzzy expression of the modus ponens (see 
\cite{belohlavek2012fuzzy} for detailed explanations).
Classic examples of such operations are G\"odel, Goguen or \L{ukasiewicz} operations.
The lattice of truth values combined with these new operations form a residuated lattice.
All these works then study the extension of the classical functional dependencies within 
their fuzzy setup, the so-called fuzzy or generalized functional dependencies.

Our approach in this paper is different though: it comes on top of the relational model 
and we do not alter its logical foundations.
Thus, it remains crisp. 
Indeed, the abstract lattices we define for each attribute are not the support for new 
logical operations, neither do they replace the underlying Boolean truth values.
Abstract values can be seen as similarity values rather than truth values.

As a consequence, we can identify some typical examples where our framework will differ 
from the fuzzy setup:
\begin{itemize}
\item here, the result of a logical operation remains true or false. 
For instance, if we have two tuples $t_1, t_2$ at hand, the result of $f_{A}(t_1[A], 
t_2[A]) = x \text{ AND } f_B(t_1[B], t_2[B]) = y$ is either \ctt{true} or \ctt{false} 
(which has in general nothing to do with values in $\Lt_A$, $\Lt_B$).
In the fuzzy set up, the result of such operation would be a value $x$ in the underlying 
residuated lattice also containing $y$ and $z$.

\item abstract FDs are different in nature from fuzzy FDs and logical expressions $x \imp 
y$ in the fuzzy set up.
An abstract FD is a lattice expression \cite{day1992lattice} over the product of the 
abstract lattices used to replace equality.
It is either true or false, depending on the closure operator induced by the data at 
hand. 
Thus, an abstract FD differs from a logical expression $x \imp y$ (possibly with hedge) 
in a residuated lattice whose truth value is another value $z$ in the residuated lattice, 
independent from the data \cite{belohlavek2012fuzzy}.
It also differs from similarity based fuzzy functional dependencies (SBFDs) since these 
latter are expressed in terms of (crisp in our case) attribute sets 
\cite{belohlavek2018relational}. 

\item Classic functional dependencies.
In our work, the validity of a functional dependency in a relation depends on the choice 
of an interpretation. 
Interpretations are a generalization of $\alpha$-cuts of fuzzy logic (in the interval 
$[0, 1]$) to any lattice.
In the fuzzy relational model, FDs are restriction of SBFDs to nonranked data tables 
equipped with equality. 
\end{itemize}

All the fuzzy models we mentioned do not consider comparability as a two-step process as 
we do in this paper.
As a consequence, they do not consider the question of different interpretations, in 
particular the problem of characterizing (strong) realities.
Similarly, they do not investigate (strongly) possible/certain functional dependencies.

At last, the fact that we rely on the classical logic leads to different 
implementation perspectives.
Our framework can be implemented as a plugin on top of a RDBMS by extending DDL and SQL 
syntax with ways to declare comparabilities and lattices. 
On the other hand, implementing the fuzzy relational model amounts instead to propose a 
new DBMS, possibly based on existing RDBMS \cite{belohlavek2007relational}.

One can note that the fuzzy set-up has also been applied to Formal Concept Analysis (FCA) 
\cite{belohlavek2012fuzzy, belohlavek2005implications, burusco1994study, 
burusco2000concept, medina2009formal, medina2010multi}.
The usual incidence relation is replaced by a (possibly residuated) lattice of truth 
values equipped with specific logical connectives.
These connectives are used to derive fuzzy concepts and fuzzy attribute implications, in 
the spirit of fuzzy functional dependencies for fuzzy relational systems 
\cite{belohlavek2018relational}.
However, since the principles are similar to fuzzy relational systems for databases, and 
our framework is database-oriented, we do not develop fuzzy FCA further.

Many contributions have been proposed to deal with uncertain FDs, among which we 
quote \cite{bertossi2013data, link2019relational, song2011differential}, and 
\cite{caruccio2015relaxed} for a survey.
The authors of \cite{koudas2006relaxing} have considered minimal relaxations for 
query conditions involving comparison operators (as selection or join queries) to obtain 
an answer.
Their works concern numerical attributes where relaxation is quantified as value 
differences, but this can also be applied to any \textit{``blackbox''} function that 
quantifies the relaxation.
Our contribution is an extension of these works within a declarative perspective, 
using lattice theory.

% ==== Conclusion ====================================================================== %

\section{Conclusion}
\label{sec:conclusion}

Defining equality between two tuples is a central problem in databases, notably due 
to data quality problems such as missing values and uncertainty.
In practice, domain experts should be able to easily give a meaningful semantics for 
equality.
In this paper, we have introduced a lattice-based formalism allowing to deal with 
many possible interpretations of data value equalities.
Our approach is able to handle both missing and uncertain values, which might be 
of practical use in applications where values are known to be imprecise and dirty. 
In order to determine whether two tuples are equal, they are first compared using a 
comparability function which returns an abstract tuple representing their similarity. 
Then, an interpretation maps this abstract tuple to a binary vector where $1$ means 
equality and $0$ difference.
We introduced realities as particular interpretations satisfying two reasonable 
consistency rules: they are increasing and guarantee that meet of two 
abstract values considered equal (interpreted as $1$) is also interpreted as $1$.
Strong realities further satisfy the rule that the upper bound of two abstract values 
considered different (interpreted as $0$) is also interpreted as $0$.

We studied this framework with respect to functional dependencies.
In this setup, we associate an abstract lattice with any relation.
This abstract lattice naturally induces abstract functional dependencies.
We showed that realities turn the abstract lattice of a relation into a closure system 
over the relation scheme.
Furthermore, we exhibited a relationship between abstract FDs and FDs based on 
realities.
On the other hand, we studied the problem of deciding whether a functional 
dependency over an incomplete database is certain (it holds in all possible worlds 
represented by the realities) or possible (there exists a possible world in which it 
holds).

Future work include both theoretical and practical contributions.
On the practical point of view, our framework could be instantiated on top of DBMS in 
order to provide a declarative view of equality.
On this purpose, it would be necessary to slightly extend SQL and DDL syntax to allow the 
user to declare comparability functions, abstract lattices and realities only on the 
relevant attributes.
In this perspective, a query relying on realities would be rewritten in order to use 
comparability functions on attributes where they are defined (otherwise, equality is 
kept).
In case a query does not use any keyword, it would be interpreted as a classical query. 
There is an ongoing work on SQL queries based on these principles within the context of a collaboration with the (\href{https://www.cemafroid.fr/index-en.htm}{CEMAFROID}). 

On the theoretical side we plan to investigate in greater depths abstract dependencies 
and their interactions with realities.
For example, if we are given a cover of an abstract lattice by abstract FDs, is it 
true that for any reality, the interpretation of this cover will contain a cover of 
the lattice interpretation?
Fascinating questions are also pending regarding possible functional dependencies.
For instance, can we decide in polynomial time that there exists a reality in which a 
given set of FDs holds?
There are also pending questions regarding functional dependencies.
In particular, the complexity of deciding whether a set of functional dependencies is 
possible remains an intriguing open question.

\subsection*{Acknowledgments}

The first author is funded by the French government IDEXISITE
initiative 16-IDEX-0001 (CAP 20-25). 
The third author is supported by the CNRS, ProFan Project.
Also, We thank the Datavalor initiative at INSA
Lyon for funding a part of this work.

\bibliographystyle{alpha}
\bibliography{biblio}

\appendix

% ---- Proofs Framework ---------------------------------------------------------------- %

\section{Proofs of Section \ref{sec:framework}}

\begin{proposition*}[\ref{prop:morphisms}]
Let $g$ be a scheme interpretation.
The following properties hold:
\begin{enumerate}
	\item $g$ increasing if and only if $h_{A_i}$ is increasing for any $1 \leq i 
	\leq n$,
	\item $g$ a $\mt$-homomorphism (resp.~$\jn$-homomorphism) if and only if for any $1 
	\leq i \leq n$, $h_{A_i}$ is a $\mt$-homomorphism (resp.~$\jn$-homomorphism) .
\end{enumerate}
\end{proposition*}

\begin{proof}
We prove the items for increasing and $\mt$-homomorphism.
Other homomorphisms and decreasing properties are obtained in a dual way.

\paragraph{Item 1} 
First, assume that for any $1 \leq i \leq n$, $h_{A_i}$ is increasing.
Let $x = \langle x_1, \allowbreak \dots, x_n \rangle$, $y = \langle y_1, 
\dots, y_n 
\rangle \in \Lt_{R}$, we have
\[
\begin{array}{r l l}
x \leq y \eqv & x_i \leq y_i \quad \text{ for any } 1 \leq i \leq n & \text{(def of 
$\Lt_{\U}$)} \\
\implies & h_{A_i}(x_i) \leq h_{A_i}(y_i) \quad \text{ for any } 1 \leq i \leq n & 
\text{($h_{A_i}$ increasing)} \\
\implies & g(x) \leq g(y) & \\
\end{array}
\]
For the only if part, assume $g$ is increasing and consider $x_i \leq y_i$ in 
$\Lt_{A_i}$. 
Build $x$ and $y$ in $\Lt_{\U}$ as follows, $x = \langle 1_1, \dots, x_i, 
\dots, 
1_n \rangle$ and $y = \langle 1_1, \dots, y_i, \dots, 1_n \rangle$.
Since $g$ is increasing, we have $g(x) \leq g(y)$ and for any $j \neq i$, 
$h_{A_j}(x_j) = h_{A_j}(y_j) = h_{A_j}(1_j) = 1$.
We must also have $h_{A_i}(x_i) \leq h_{A_i}(y_i)$.
Because for any pair $x_i \leq y_i$ in $\Lt_{A_i}$ and for any $1 \leq i \leq n$ we 
can build 
such $x, y$ showing $h_{A_i}(x_i) \leq h_{A_i}(y_i)$, we have that $h_{A_i}$ 
is 
increasing for any $1 \leq i \leq n$.
		
\paragraph{Item 2} 
Let us assume that for each $1 \leq i \leq n$, $h_{A_i}$ is a $\mt$-homomorphism.
For $x, y \in \Lt_{\U}$, one has:
\begin{align*}
g(x \mt y) = & \, \langle h_{A_1}(x_1 \mt y_1), \dots, h_{A_n}(x_n \mt y_n) 
\rangle \quad \text{(def of $g$)} \\
= & \, \langle h_{A_1}(x_1) \mt h_{A_1}(y_1), \dots, h_{A_n}(x_n) \mt 
h_{A_n}(y_n) \rangle \quad \text{($\mt$-homomorphism)} \\
= & \, \langle h_{A_1}(x_1), \dots, h_{A_n}(x_n) \rangle \mt \langle 
h_{A_1}(y_1), \dots, h_{A_n}(y_n) \rangle \quad \text{(in $\{0, 1\}^n$)} \\
= & \, g(x) \mt g(y) \\
\end{align*}
As for the only if part, we use contrapositive. 
Assume without loss of generality that $h_{A_1}$ is not $\mt$-homomorphic, and let $x_1, 
y_1 \in 
\Lt_{A_1}$ 
such that $h_{A_1}(x_1 \mt y_1) \neq h_{A_1}(x_1) \mt h_{A_1}(y_1)$. 
Let $x = \langle x_1, \dots, x_n \rangle$ and $y = \langle y_1, \dots, y_n 
\rangle$. 
Then, $g(x) \mt g(y) \neq g(x \mt y)$ by definition of $g$, as for $A_1$ we 
have 
$h_{A_1}(x_1 \mt y_1) \neq h_{A_1}(x_1) \mt h_{A_1}(y_1)$.
\end{proof}

% ---- Proofs FDs ---------------------------------------------------------------------- %

\section{Proofs of Section \ref{sec:fd}}

%The main result can now be given.

\begin{proposition*}[\ref{prop:meet-sublattice}]
Let $r$ be a relation over a scheme context $\C_{\U}$.
Then, $\Lt_r \in \Sb(\Lt_{\U})$.
\end{proposition*}

\begin{proof}
First, we show that $\Lt_r$ is indeed a lattice.
It is sufficient to prove that it has a well-defined $\mt$ operation and a top element 
\cite{davey2002introduction}.
Also following \cite{davey2002introduction}, we have that $\bigwedge \emptyset = \langle 
1_{A_1}, \dots, 1_{A_n} \rangle$, the top element of $\Lt_{\U}$.
Thus, $\Lt_r$ has a top element.
Now we prove that $\Lt_r$ has a $\mt$ operation.
Since $\Lt_r \subseteq \Lt_{\U}$, it is sufficient to prove that for every $t_1, t_2 \in 
\Lt_r$, $t_1 \mt t_2 \in \Lt_r$, where $\mt$ is the meet operation of $\Lt_{\U}$.
By definition of $\Lt_r$, $t_1 = \bigwedge T_1$ and $t_2 = \bigwedge T_2$ for some $T_1, 
T_2 \subseteq f_{\U}(r)$.
Since, $t_1 \mt t_2 = \bigwedge T_1 \mt \bigwedge T_2 = \bigwedge T_1 \cup T_2$ and $T_1 
\cup T_2 \subseteq f_{\U}(r)$, $t_1 \mt t_2 \in \Lt_r$ holds, by construction of $\Lt_r$.
Consequently, $\Lt_r$ is indeed a lattice.
As $\Lt_r \subseteq \Lt_{\U}$ and $\Lt_r$ is closed for the $\mt$ operation, it is 
furthermore a $\mt$-sublattice of $\Lt_{\U}$.
\end{proof}

\begin{proposition*}[\ref{prop:sat-fd-reality-1}]
Let $r$ be a relation over $\C_{\U}$, $g$ a reality, and $x \imp y$ an abstract 
FD.
We have $r \models_g g(x) \imp g(y)$ if and only if for any $z \in \Lt_r$, $g(x) 
\subseteq g(z)$ implies $g(y) \subseteq g(z)$, denoted $\Lt_r \models_g g(x) \imp g(y)$.
\end{proposition*}

\begin{proof}
The if part is clear.
We prove the only if part.
Let $z \in \Lt_r$. 
If $z \in f_{\U}(r)$ then $g(x) \subseteq g(z)$ implies $g(y) \subseteq g(z)$ by 
assumption.
Hence consider an element $z \in \Lt_r \setminus f_{\U}(r)$ such that $g(x) 
\subseteq 
g(z)$.
By definition of $\Lt_r$, $z = \Mt F_z$ where $F_z = \{m \in f_{\U}(r) \mid z 
\leq m 
\}$.
As $g$ is increasing, it must be that for all $m$ in $F_z$, $g(x) \subseteq g(m)$.
Since $r \models_g g(x) \imp g(y)$, we also have that $g(y) \subseteq g(m)$ for 
any 
$m 
\in F_z$, and in particular $g(y) \subseteq \bigcap_{m \in F_z} g(m)$.
Moreover, $g$ is a $\mt$-homomorphism so that $z = \Mt F_z$ implies $g(z) = g(\Mt 
F_z) = 
\bigcap_{m \in F_z} g(m)$.
Therefore $g(y) \subseteq g(z)$ and $\Lt_r \models_g g(x) \imp g(y)$.
\end{proof}

\begin{lemma*}[\ref{lem:projection-g}]
Let $r$ be a relation over a scheme context $\C_{\U}$.
Let $x \in \Lt_{\U}$, and $g$ be a reality. 
Then $g(\cl(\pi_g(x))) = \cl_g(g(\pi_g(x))) = \cl_g(g(x))$.
\end{lemma*}

\begin{proof}
The equality $\cl_g(g(\pi_g(x))) = \cl_g(g(x))$ directly follows by construction of 
$\pi_g(x)$ as $g(\pi_g(x)) = g(x)$.
As $\pi_g(x) \leq \cl(\pi_g(x))$ in $\Lt_{\U}$ and $g$ is increasing, we have that 
$g(x) = g(\pi_g(x)) \subseteq \cl_g(g(\pi_g(x)))$.
To show that $g(\cl(\pi_g(x))) = \cl_g(g(x))$, we prove that $\cl(\pi_g(x)) = 
\csf{min}_{\leq}\{y \in \Lt_r \mid g(x) \subseteq g(y)\}$, which implies that 
$g(\cl(\pi_g(x))) = \csf{min}_{\subseteq}\{F \in g(\Lt_r) \mid g(x) \subseteq F\} = 
\cl_g(g(x))$.
If $\cl(\pi_g(x))$ is the bottom element of $\Lt_r$, $g(\cl(\pi_g(x)))$ must be the 
bottom element of $g(\Lt_r)$ and $g(\cl(\pi_g(x))) = \cl_g(g(x))$ is clear.
If this is not the case, let $y \ngeq \cl(\pi_g(x))$, $y \in \Lt_r$.
By closure properties, $y \ngeq \cl(\pi_g(x))$ is equivalent to $y \ngeq \pi_g(x)$.
To satisfy $y \ngeq \pi_g(x)$, there must be $A \in \U$ such that $\pi_g(x)[A] 
\nleq_A y[A]$.
We have two cases: $A \in g(x)$ or $A \notin g(x)$.
If $A \notin g(x)$, we have that $A \notin g(\pi_g(x))$ as $g(x) = g(\pi_g(x))$.
Consequently, $\pi_g(x)[A] = 0_A$ by construction of $\pi_g(x)$ and hence 
$\pi_g(x)[A] \leq_A y[A]$.
Therefore, it must be that $A \in g(x)$.
As $\pi_g(x)[A] \nleq_A y[A]$ by assumption, it follows that $A \notin g(y)$ and  
therefore that $g(x) \nsubseteq g(y)$.
Consequently, for any $y \ngeq \pi_g(x)$, $y \in \Lt_r$, we have $g(x) \nsubseteq 
g(y)$ so that $g(\cl(\pi_g(x))) = \csf{min}_{\subseteq}\{F \in g(\Lt_r) \mid g(x) 
\subseteq F\} = \cl_g((x))$, concluding the proof.
\end{proof}

\begin{lemma*}[\ref{lem:certain-FD}]
Let $r$ be a relation over $\C_{\U}$, and $X \imp A$ a functional dependency.
Then:
\begin{enumerate}
	\item $X \imp A$ is certain if and only if for any $x \in 
	\Lt_r$, either $x[A] = 1_A$ 
	or there exists $B \in X$ such that $x[B] = 0_B$.
	\item $X \imp A$ is strongly certain if and only if for any $x \in \Lt_r$
	either $x[A] \geq_A \Jn \coPm(\Lt_A)$ or there exists $B \in X$ such that 
	$x[B] 
	\ngeq_B c$ for any $c \in \coPm(\Lt_B)$.
\end{enumerate}
\end{lemma*}

\begin{proof}
We prove the items in order.
\paragraph{Item 1} We begin with the if part.
There are two cases, either $X = \emptyset$, or $X \neq \emptyset$.
In the first case, our assumption implies that $x[A] = 1_A$ for any $x \in 
\Lt_r$.
Hence, by construction of a reality $g$, $A \in g(x)$ for any $x \in \Lt_r$, 
in 
particular the bottom element of $\Lt_r$, and $r \models_g X \imp A$ holds.
Assume $X \neq \emptyset$.
Let $x \in \Lt_r$, and $g$ be any reality.
If $x[A] = 1_A$, then $A \in g(x)$ since $h_A(1_A)$ must be $1$ by definition 
of 
an interpretation.
Hence if $X \subseteq g(x)$, we still have $A \in g(x)$ and $x$ satisfies $X 
\imp 
A$ 
under $g$.
Suppose now $x[A] \neq 1_A$.
By assumption, there exists $B \in X$ such that $x[B] = 0_B$. 
Hence $h_B(x[B]) = 0$, and $B \notin g(x)$ so that $X \nsubseteq g(x)$.
Therefore, in any case, $x$ satisfies $X \imp A$ through $g$ and $r 
\models_{g} X 
\imp A $ for any $g \in \cc{R}$.

We prove the only if part using contrapositive.
Assume there exists $x \in \Lt_r$ such that for all $B \in X$, $x[B] \neq 
0_B$ 
and 
$x[A] \neq 1_A$. 
We construct a reality $g$ such that $r \not\models_{g} X \imp A$:
\begin{itemize}
	\item for $A$ put $h_A(y_A) = 1$ if $y_A \geq z_A$ in $\Lt_A$ with $z_A$ 
	some 
	cover of $x[A]$.
	Note that such a cover must exist since $x[A] \neq 1_A$.
	Put $h_A(y_A) = 0$ otherwise, so in particular $h_A(x[A]) = 0$.
	\item For any $B$ in $X$, let $h_B(y_B) = 1$ if $y_B \geq x[B]$ in 
	$\Lt_B$ 
	and 
	$h_B(y_B) = 0$ otherwise.
	Since $x[B] \neq 0_B$, $h_B(0_B) = 0$ is satisfied.
	\item For any other attribute $C \in \U$, let $h_C(y_C) = 1$ if $y_C = 
	1_C$ 
	and $0$ 
	otherwise.
\end{itemize}
It is clear that any $h$ thus defined is a $\mt$-homomorphism.
Using Proposition \ref{prop:meet-sublattice} and Definition \ref{def:realities}, the 
mapping $g$ obtained by combining all attributes interpretations is a reality.
If $X = \emptyset$, by construction of $g$ and $x$, $A \notin g(x)$ and hence 
$r \not\models_g X \imp A$.
Assume $X \neq \emptyset$.
For any $B \in X$, one has $B \in g(x)$ as $g_{\mid B}(x) = h_B(x[B]) = 1$ by 
definition 
of $h_B$.
We also have $A \notin g(x)$ since $g_{\mid A}(x) = h_A(x[A]) = 0$ by 
construction of 
$h_A$.
Therefore $r \not\models_{g} X \imp A$, and $X \imp A$ is not certain, 
concluding 
the 
proof of this item. 

\paragraph{Item 2} We begin with the if part.
Let $g$ be a strong reality and let $x \in \Lt_r$.
If $X \nsubseteq g(x)$, then $g(x)$ vacuously satisfies $X \imp A$.
Let us assume that $X \subseteq g(x)$ so that for any $B \in X$, there is $c 
\in 
\coPm(\Lt_B)$ such that $x[B] \geq c$.
As $g$ is a strong reality, there must be some $c' \in \coPm(\Lt_A)$ such 
that 
$g_{\mid A} = h_A^{c'}$ by Proposition \ref{prop:coprime-prime}.
By assumption, we have that $x[A] \geq c'$ as $x[A] \geq \Jn \coPm(\Lt_A)$.
Therefore, $A \in g(x)$ and we conclude that $r \models_g X \imp A$ for any 
reality $g 
\in \cc{R}_s$.
Note that in the case where $X = \emptyset$, $X \subseteq g(x)$ holds for any 
$x 
\in 
\Lt_r$ so that for all such $x$, we have $x[A] \geq \Jn \coPm(\Lt_A)$. 

We move to the only if part.
Let us assume that $X \imp A$ is strongly certain, and assume that there 
exists 
$x \in 
\Lt_r$ such that for any $B \in X$, $x[B] \geq c$ for some $c \in 
\coPm(\Lt_B)$.
If $X = \emptyset$, this is the case for any $x \in \Lt_r$.
We put $k = \card{\coPm(\Lt_A)}$.
We consider the sequence $g_1, \dots, g_k$ of realities as follows, for $1 
\leq i 
\leq 
k$:
\begin{itemize}
	\item $g_{i_{\mid B}} = h_B^c$ for some $c \in \coPm(\Lt_B)$ such that 
	$x[B] 
	\geq 
	c$, 
	$B \in X$,
	\item $g_{i_{\mid A}} = h_A^c$ for some $c \in \coPm(\Lt_A)$ such that 
	$g_{i_{\mid 
			A}}	\neq g_{j_{\mid A}}$ for any $1 \leq j < i$.
	\item $g_{i_{\mid C}} = h_C^c$ for some $c \in \coPm(\Lt_C)$, $C \in \U 
	\setminus X 
	\cup \{A\}$
\end{itemize}
For any $1 \leq i \leq k$, we have that $X \subseteq g_k(x)$.
As $X \imp A$ is strongly certain, we must also have $A \in g_k(x)$ by Definition
\ref{def:fd-glr}.
By construction of the sequence $g_1, \dots, g_k$, for any $c \in 
\coPm(\Lt_A)$ 
there 
is a (unique) $1 \leq i \leq k$ such that $g_{i_{\mid A}} = h_A^c$.
However, as $A \in g_i(x)$ for any $g_i$, we have that $x[A] \geq c$ for any 
$c 
\in 
\coPm(\Lt_A)$ and hence $x[A] \geq \Jn \coPm(\Lt_A)$.
Since this reasoning can apply to any $x \in \Lt_r$, the second item follows. 
\end{proof}

\begin{lemma*}[\ref{lem:possible-FD}]
Let $r$ be a relation over $\C_{\U}$, and $X \imp A$ a functional dependency.
Then:
\begin{enumerate}
	\item $X \imp A$ is not possible if and only if there exists $x \in \Lt_r$, 
	such that $x[A] = 0_A$ and $x[B] = 1_B$ for any $B \in X$.
	\item $X \imp A$ is strongly possible if and only if there exists $c_X \in 
	\Lt_{\U}$ such that $c_X[B] = c$, for some $c$ in $\coPm(\Lt_B)$, for every 
	$B \in X$, $c_X[B] = 0_B$ for every $B \notin X$, and 
	such that $\cl(c_X)[A] \geq c$ for some $c \in \coPm(\Lt_A)$.
\end{enumerate}
\end{lemma*}

\begin{proof}
We prove the statements in order. Let $r$ be a relation over $\C_{\U}$ and $X 
\imp A$ 
a functional dependency.
For the first item we need to introduce particular realities.
Recall that $\cc{A}(\Lt_A)$ is the set of atoms of $\Lt_A$.
Let $a \in \cc{A}(\Lt_A)$.
We define a reality $g^a$ as follows:
\[
g_{\mid B}^a(x) = \left\{
\begin{array}{l l}
	1 & \text{if } x[B] = 1_B \\
	0 & \text{otherwise}
\end{array}
\right.
\; \text{for any $B \neq A$ and} \;
g_{\mid A}^a(x) = \left\{
\begin{array}{l l}
	1 & \text{if } x[A] \geq a \\
	0 & \text{otherwise}
\end{array}
\right.
\]

\paragraph{Item 1} We begin with the if part.
Hence, assume there exists $x \in \Lt_r$ such that $x[A] = 0_A$ and for every 
$B 
\in X$, 
$x[B] = 1_B$.
In the case where $X = \emptyset$, then $x$ just satisfies $x[A] = 0_A$.
Let $g$ be a reality.
By definition of interpretations, we must have $A \notin g(x)$ and  $X 
\subseteq g(x)$.
Hence, $x$ does not satisfy $X \imp A$, so that $r \not\models_{g} X \imp A$.
Since this holds for any reality $g$, we conclude that $X \imp A$ is not 
possible.

We show the only if part.
Assume that $X \imp A$ is not possible in $\Lt_r$, that is for any reality 
$g$, 
$\Lt_r 
\not\models_g X \imp A$.
In particular, this is true for any reality $g^a$, $a \in \cc{A}(\Lt_A)$.
Then, for any $a \in \cc{A}(\Lt_A)$, there exists $x_a \in \Lt_r$ such that 
$x[B] 
= 1_B$ 
for any $B \in X$ and $x[A] \ngeq_A a$, by definition of $g^a$ and 
\ref{def:fd-glr}.
As $\Lt_r$ is a meet-sublattice of $\Lt_{\U}$, $x = \Mt_{a \in \cc{A}(\Lt_A)} 
x_a$ also 
belongs to $\Lt_r$.
However, it must be that $x[B] = 1_B$ for any $B \in X$.
Furthermore, we have that $x[A] = 0_A$ as $x_a \ngeq_A a$ for any $a \in 
\cc{A}(\Lt_A)$,
which concludes the proof of the item.

\paragraph{Item 2} We begin with the only if part.
Let $g$ be a strong reality such that $r \models_g X \imp A$ and consider the 
following 
$c_X$:
\begin{itemize}
	\item $c_X[B] = 0_B$ for $B \notin X$,
	\item $c_X[B] = c$ for $c \in \coPm(\Lt_B)$ such that $g_{\mid B} = 
	h_B^c$ 
	for any $B \in X$.
	Such a $c$ exists by Proposition \ref{prop:coprime-prime} as $g$ is a strong 
	reality.
\end{itemize}
Again, if $X = \emptyset$, then $c_X[B] = 0_B$ for every $B \in \U$.
We consider $\cl(c_X) \in \Lt_r$.
As $c_X \leq \cl(c_X)$, and by construction of $c_X$, we must have $X \subseteq 
g(\cl(c_X))$.
However, $r \models_g X \imp A$ which implies that $A \in g(\cl(c_X))$ by Definition
\ref{def:fd-glr}.
Hence, since $g$ is a strong reality, there must be some $c \in \coPm(\Lt_A)$ 
such that 
$g_{\mid A} = h_A^c$ and $\cl(c_X)[A] \geq c$.

We move to the if part.
Assume such $c_X$ exists.
We construct a strong reality $g$ such that $r \models_g X \imp A$ as follows:
\begin{itemize}
	\item $g_{\mid B} = h_B^c$ for some $c \in \coPm(\Lt_B)$ if $B \notin X 
	\cup 
	\{A\}$,
	\item $g_{\mid B} = h_B^{c_X[B]}$ if $B \in X$,
	\item $g_{\mid A} = h_A^c$ for some $c \in \coPm(\Lt_A)$ such that $c 
	\leq 
	\cl(c_X)[A]$.
	Note that such $c$ must exist by assumption.
\end{itemize}
Thus defined, $g$ is a strong reality as $c_X[B] \in \coPm(\Lt_B)$ for any $B 
\in X$.
Now we show that $r \models_g X \imp A$.
Let $x \in \Lt_r$.
Either $x \geq \cl(c_X)$ or $x \ngeq \cl(c_X)$.
In the first case, by construction of $g$ we have $X \subseteq g(x)$, but we 
also 
have 
$A 
\in g(x)$ as $g_{\mid A} = h_A^c$ with $c \leq \cl(c_X)[A]$, so that $X \imp A$ 
is satisfied for $g(x)$.
Now suppose $x \ngeq \cl(c_X)$.
We have that $x \ngeq c_X$.
However for any $B \notin X$, $c_X[B] = 0_B$ so that necessarily $c_X[B] \leq 
x[B]$.
Therefore, there must exist $B \in X$ such that $x[B] \ngeq c_X[B]$.
Hence, by construction of $g$ we will obtain $B \notin g(x)$ and consequently 
$X 
\nsubseteq g(x)$ so that $X \imp A$ is also satisfied by $g(x)$.
As for any $x \in \Lt_r$, $X \subseteq g(x)$ implies $A \in g(x)$, we 
conclude that $r \models_g X \imp A$ and $X \imp A$ is strongly possible.
\end{proof}

% ---- More results on AFDs ------------------------------------------------------------ %

\section{Computational results on abstract FDs} 
\label{sec:AFDs}

In this section, we detail some computational characteristics of abstract FDs.
First, we recall some notations to ease the understanding.
Let $\C_{\U}$ be a scheme context with comparability function $f_{\U}$ and corresponding 
lattice $\Lt_{\U}$ of abstract tuples.
Let $\Lt \in \Sb(\Lt_{\U})$ with associated closure operator $\cl$ and let $x \imp y$ be 
an abstract FD.
We write $\Lt \models x \imp y$ if for every $z \in \Lt$, $x \leq z$ implies $y \leq z$, 
that is $z \models x \imp y$.
We have $\Lt \models x \imp y$ if and only if $y \leq \cl(x)$.
If $\IS$ is a set of abstract FDs, $\Lt \models \IS$ if $\Lt \models x \imp y$ for every 
$x \imp y$ in $\IS$.

Let $r$ a relation over $\C_{\U}$.
The relation $r$ is associated to a lattice $\Lt_r \in \Sb(\Lt_{\U})$ given by 
\[ 
\Lt_r = \left(\left\{ \bigwedge T \mid T \subseteq f_{\U}(r)\right\}, \leq_{\U}\right)
\]
The closure operator associated to $\Lt_r$ is $\cl_r$.
We write $r \models x \imp y$ if for every $z \in f_{\U}(r)$, $z \models x \imp y$.
The notation $r \models \IS$ follows.
Recall from Proposition 3 in the manuscript that $r \models x \imp y$ if and only if $\Lt_r 
\models x \imp y$.

Similarly to $r$ and $\Lt_r$, a set $\IS$ of abstract FDs is associated to a lattice 
$\Lt_{\IS}$ defined by $\Lt_{\IS} = (\{x \in \Lt_{\U} \mid x \models \IS\}, \leq_{\U})$.
We call the corresponding closure operator $\cl_{\IS}$.

We study the implication problem for abstract functional dependencies: \textit{``given a 
$\IS$, does the abstract FD $x \imp y$ follow from $\IS$?''}
On this purpose, we use the extended Armstrong axioms for lattice implications 
\cite{day1992lattice}.
They are straightforward extensions of Armstrong axioms for functional dependencies.
They read as follows, for every $x, y, z \in \Lt_{\U}$:
\begin{enumerate}
\item if $x \leq y$, then $\IS \vdash_{\C_{\U}} x \rightarrow y$ (reflexivity);
\item if $\IS \vdash_{\C_{\U}} x \rightarrow y$ then $\IS \vdash_{\C_{\U}} x \jn z
\rightarrow y \jn z$ (augmentation);
\item if $\IS \vdash_{\C_{\U}} x \rightarrow y$ and $\Sigma \vdash_{\C_{\U}} y 
\rightarrow z$ then $\Sigma \vdash_{\C_{\U}} x \rightarrow z$ (transitivity).
\end{enumerate}

\begin{definition}
$\IS \vdash_{\C_{\U}} x \rightarrow y$ if there exists a derivation (or a proof) of  
$x \rightarrow y$  by using the extended Armstrong axioms.
\end{definition}

\begin{definition}
$\IS \models_{\C_{\U}} x \rightarrow y$ if for every relation $r$, $r 
\models_{\C_{\U}} \IS$ entails $r \models_{\C_{\U}} x \rightarrow y$.  
\end{definition}

If no confusion arises, we drop the subscript $\C_{\U}$.
The next propositions establish useful properties for subsequent discussions.

\begin{proposition} \label{prop:rep-rmodels}
We have $r \models \IS$ if and only if $\Lt_r \subseteq \Lt_{\IS}$.
\end{proposition}
\begin{proof}
We begin with the only if part.
Assume that $r \models \IS$. 
Every abstract tuple $x \in f_{\U}(r)$ satisfies $\IS$, so that $x \in \Lt_{\IS}$ holds 
by definition of $\Lt_{\IS}$.
Since $\Lt_r = (\{ \bigwedge T \mid T \subseteq f_{\U}(r)\}, \leq_{\U})$, $\Lt_{\IS} \in 
\Sb(\Lt_{\U})$ and $f_{\U}(r) \subseteq \Lt_{\IS}$, $\Lt_r \subseteq 
\Lt_{\IS}$ 
follows.

We move to the if part.
We use contrapositive.
Assume that $r \not\models \IS$.
There exists an abstract tuple $x \in f_{\U}(r)$ such that $x \not\models \IS$.
Thus, $x \notin \Lt_{\IS}$ and $\Lt_r \nsubseteq \Lt_{\IS}$ follows, which concludes the 
proof.
\end{proof}

The following proposition is a direct consequence of Theorem 3.1 and Theorem 3.2 in 
\cite{day1992lattice}.

\begin{proposition} \label{prop:rep-sigmaproof}
We have $\IS \vdash x \imp y$ if and only if $\Lt_{\IS} \models x \imp y$.
\end{proposition}

Thus, there is a strong relationship between a set of abstract FDs and its 
models in the lattice $\Lt_{\U}$.
We show that deciding whether $\IS \vdash x \imp y$ can be decided in polynomial time, 
using Proposition \ref{prop:rep-sigmaproof} and Algorithm \ref{alg:rep-closure}.
Recall that the scheme context $\C_{\U}$, including each abstract lattice, is part of our 
input. 
First, we prove that Algorithm \ref{alg:rep-closure} correctly computes $\phi_{\IS}(x)$ 
for every $x \in \Lt_{\U}$ in polynomial time.

\begin{theorem} \label{thm:rep-algo-closure}
Let $\C_{\U}$ be a scheme context, $\IS$ a set of abstract FDs and $x$ an abstract tuple.
Then, Algorithm \ref{alg:rep-closure} computes $\phi_{\IS}(x)$ in polynomial time in the 
size of $\IS$ and $\C_{\U}$.
\end{theorem}

\begin{proof}
First we show the correctness of the algorithm.
Observe that Algorithm \ref{alg:rep-closure} always terminates.
Let $x \in \Lt_{\U}$ and consider the output $x^{\IS}$ of the algorithm.
We show that $x^{\IS} = \phi_{\IS}(x)$.
We first prove that $x^{\IS}$ is a model of $\IS$ and that $\phi_{\IS}(x) \leq x^{\IS}$.
Suppose for contradiction that $x^{\IS} \not\models \IS$.
Then, there exists an abstract FD $y \imp z$ in $\IS$ such that $y \leq x^{\IS}$ but $z 
\nleq x^{\IS}$.
However, this contradicts the fact that Algorithm \ref{alg:rep-closure} ends on $x^{\IS}$.
Thus, $x^{\IS} \models \IS$ and since $x \leq x^{\IS}$, we deduce that $\phi_{\IS}(x) 
\leq \phi_{\IS}(x^{\IS}) = x^{\IS}$.

Now we show that $x^{\IS} \leq \phi_{\IS}(x)$.
We proceed by induction on the while condition of the algorithm.
At step $0$, $x^{\IS} = x$ and $x^{\IS} \leq \phi_{\IS}(x)$ follows from the extensivity 
of $\phi_{\IS}$.
let us assume that up to a given step $i$, $x^{\IS} \leq \phi_{\IS}(x)$ holds.
If there is no abstract FD $y \imp z$ such that $y \leq x^{\IS}$ then the result follows.
Now, suppose that there is an abstract FD $y \imp z$ in $\IS$ such that $y \leq x^{\IS}$.
Since $y \imp z$ holds in $\IS$, we derive that $z \leq \phi_{\IS}(y)$.
Thus, we have $z \leq \phi_{\IS}(y) \leq \phi_{\IS}(x^{\IS}) \leq \phi_{\IS}(x)$ by 
inductive hypothesis, and using the properties of $\phi_{\IS}$.
Consequently, $z \jn x^{\IS} \leq \phi_{\IS}(x)$ holds.
Hence, we derive by induction that the output $x^{\IS}$ of Algorithm 
\ref{alg:rep-closure} satisfies $x^{\IS} \leq \phi_{\IS}(x)$.
Consequently, $x^{\IS} = \phi_{\IS}(x)$ indeed holds.

The total complexity of Algorithm \ref{alg:rep-closure}, depends on the cost of checking 
conditions of the while loop and the cost of the join operation.
Since the abstract lattices of each attribute context are part of the input, this 
operation can be done in polynomial time in the size of $\IS$ and $\C_{\U}$, concluding 
the proof.
\end{proof}

\begin{algorithm2e}[h!]
\caption{Closure algorithm}
%	\SetVline

\KwIn{A set $\IS$ of abstract FDs over $\C_{\U}$, $u \in \Lt_{\U}$}
\KwOut{The closure  $x^\IS$ of $x$.}

\Begin{
$x^\Sigma = x$\;		
\While{there exists  $y \imp z \in \IS$ such that $z \leq  x^\IS$}{
		$x^\IS = x^{\IS} \jn z$\;
		$\IS = \IS \setminus \{y \imp z\}$\;}		
}	
return $x^\IS$ \;
\label{alg:rep-closure}	
	
\end{algorithm2e}

Remark that Algorithm \ref{alg:rep-closure} is a straightforward extension of classical 
closure algorithms \cite{beeri1979computational}.
According to Proposition \ref{prop:rep-sigmaproof}, $\IS \vdash x \imp y$ can be decided 
by computing $\phi_{\IS}(x)$ with Algorithm \ref{alg:rep-closure}, and then testing that 
$y \leq \phi_{\IS}(x)$.
We deduce the following theorem.

\begin{theorem}
Let $\C_{\U}$ be a scheme context, $\IS$ a set of abstract FDs and $x \imp y$ an abstract 
FD.
Testing that $\IS \vdash x \imp y$ can be done in polynomial time in the size of $\IS$ 
and $\C_{\U}$.
\end{theorem}

Then, we demonstrate that $\IS \vdash x \imp y$ 
(equivalently that $y \leq \phi_{\IS}(x)$), is not equivalent to $\IS \models x \imp y$.
More precisely, we show that extended Armstrong axioms are sound but incomplete.
This is because comparability functions are not required to be reflexive.
% \remove{on 
% the \ctt{null} value}.

\begin{lemma} \label{lem:rep-sound}
Let $\C_{\U}$ be a scheme context, $\IS$ a set of abstract FDs over $\C_{\U}$ and $x \imp 
y$ another abstract FD.
Then, $\IS \vdash x \imp y$ implies $\IS \models x \imp y$.
\end{lemma}

\begin{proof}
Let $\C_{\U}$ be a scheme context, $\IS$ a set of abstract FDs and $x \imp y$ an abstract 
FD such that $\IS \vdash x \imp y$.
From Proposition \ref{prop:rep-sigmaproof}, we have $\Lt_{\IS} \models x \imp y$.
Let $r$ be a relation over $\C_{\U}$ such that $r \models \IS$.
By Proposition \ref{prop:rep-rmodels}, $\Lt_r \subseteq \Lt_{\IS}$ holds.
As $\Lt_r \subseteq \Lt_{\IS}$, $\Lt_r \models x \imp y$ holds and so does $r \models x 
\imp y$, which concludes the proof.
\end{proof}

\begin{lemma} \label{lem:rep-complete}
There exists a scheme context $\C_{\U}$, an abstract FD $x \imp y$ and a set $\IS$ of 
abstract FDs such that $\IS \models x \imp y$ but $\IS \not\vdash x \imp y$.
\end{lemma}

\begin{proof}
Let $\U = \{A, B\}$ be a relation scheme where $\dom(A) = \dom(B) = \cb{N} \cup 
\{\ctt{null}\}$.
We associate the abstract lattice $\Lt_A$ and $\Lt_B$ presented in Figure 
\ref{fig:completeness-1} to $A$ and $B$ (resp.).
The lattice $\Lt_R = \Lt_A \times \Lt_B$ is the lattice of all possible abstract tuples.
We also define two comparability functions $f_A, f_B$ for $A, B$ (resp.) as follows:
\[ 
f_A(u, v) = f_B(u, v) = \begin{cases}
1 & \text{if } u = v \text{ with } u \neq \ctt{null} \text{ or } v \neq \ctt{null} \\
u & \text{if } u = v = \ctt{null} \\
0 & \text{otherwise.} 
\end{cases}
\]
We consider the scheme context $\cc{C}_{\U} = \{(A, f_A, \Lt_A), (B, f_B, \Lt_B)\}$.
\begin{figure}[ht!]
\centering 
\includegraphics[scale=1.0, page=1]{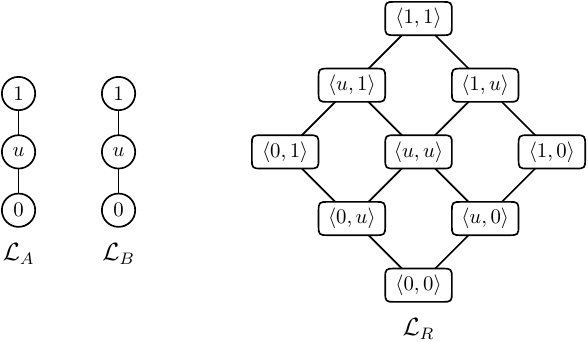}%
\caption{Abstract lattices}
\label{fig:completeness-1}
\end{figure}
Let $\IS = \{ \langle 0, 1 \rangle \imp \langle u, 1 \rangle, \allowbreak \langle 1, 0 \rangle 
\allowbreak \imp \langle 1, 1 \rangle\}$ be a set of abstract FDs.
The abstract lattice $\Lt_{\IS}$ associated to $\IS$ is given in Figure 
\ref{fig:completeness-2}.
\begin{figure}[ht!]
\centering 
\includegraphics[scale=1.0, page=2]{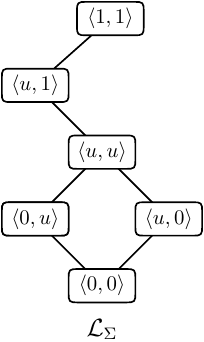}%
\caption{Abstract lattice $\Lt_{\IS}$}
\label{fig:completeness-2}
\end{figure}
Using Proposition \ref{prop:rep-sigmaproof} and the fact that $\Lt_{\IS} \models x \imp 
y$ if and only if $y \leq \phi_{\IS}$, we derive that $\IS \not\vdash \langle 0, u 
\rangle \imp \langle u, 0 \rangle$ as $\langle 0, u \rangle \in \Lt_{\IS}$.

We show that there is no relation $r$ over $\C_{\U}$ satisfying $r \models 
\IS$ and $r \not\models \langle 0, u \rangle \imp \langle u, 0 \rangle$.
It follows that $\IS \models \langle 0, u \rangle \imp \langle u, 0 \rangle$.
Assume for contradiction such a $r$ exists.
By Proposition \ref{prop:rep-rmodels}, $r \models \IS$ is equivalent to $\Lt_r \subseteq 
\Lt_{\IS}$.
Thus, $r \not\models \langle 0, u \rangle \imp \langle u, 0 \rangle$ implies that 
$\langle 0, u \rangle \in \Lt_r$ since it is the unique counter-example to $\langle 0, u 
\rangle \imp \langle u, 0 \rangle$ in $\Lt_{\IS}$.
Moreover, $\langle u, 0 \rangle$ is a meet-irreducible element of $\Lt_{\IS}$. 
It follows that $\langle u, 0 \rangle \in f_{\U}(r)$ must hold.
By definition of $f_A, f_B$, we deduce that $r$ contains at least two tuples $t_1, t_2$ 
such that $t_1(A) = i$, $t_2(A) = j$ and $t_1(B) = t_2(B) = \ctt{null}$ where $i, j$ are 
distinct integers.
We have $f_{\U}(t_1, t_2) = \langle 0, u \rangle$ as required.
However, we also have $f_{\U}(t_1, t_1) = \langle 1, u \rangle$ with $\langle 1, u 
\rangle \notin \Lt_{\IS}$.
As a consequence, $\Lt_r \nsubseteq \Lt_{\IS}$ from which $r \not\models \IS$ follows by 
Proposition \ref{prop:rep-rmodels}, a contradiction.
\end{proof}

Using Lemma \ref{lem:rep-sound} and Lemma \ref{lem:rep-complete}, we obtain the following 
theorem.

\begin{theorem} \label{thm:rep-axioms}
The extended Armstrong axioms are sound but not complete.
\end{theorem}

As a consequence, we have yet no guarantees on the complexity of deciding $\IS \models x 
\imp y$.
However, enforcing reflexivity for each comparability functions solves the problem.

\begin{lemma}\label{lem:rep-ref-complete}
Let $\C_{\U}$ be a scheme context where for each attribute $A$ in $\U$, $f_A(u, u) = 1_A$ 
for every $u \in \dom(A)$.
Let $\IS$ a set of abstract FDs over $\C_{\U}$ and $x \imp y$ another abstract FD.
Then, $\IS \models x \imp y$ implies $\IS \vdash x \imp y$.
\end{lemma}

\begin{proof}
We show the contrapositive. 
Assume that $\IS \not\vdash x \imp y$. 
We construct a relation $r$ such that $r \models \IS$ but $r \not\models x \imp y$.
Using Proposition \ref{prop:rep-sigmaproof}, $\IS \not\vdash x \imp y$ is equivalent to 
$y \nleq \phi_{\IS}(x)$.
Thus, let $r = \{t_1, t_2\}$ be a relation over $\C_{\U}$ such that $f_{\U}(t_1, t_2) = 
\phi_{\IS}(x)$.
Such a relation must exist as comparability functions are surjective by definition.
Now, as each comparability function is assumed to be reflexive, we have $f_{\U}(t_1, 
t_1) = f_{\U}(t_2, t_2) = 1_{\U}$ where $1_{\U}$ is the top element of $\Lt_{\U}$.
Because $\Lt_{\IS} \in \Sb(\Lt_{\U})$ and $\phi_{\IS}(x) \in \Lt_{\IS}$, we deduce that 
$\Lt_r \subseteq \Lt_{\IS}$.
By Proposition \ref{prop:rep-rmodels}, this is equivalent to $r \models \IS$.
Now since $y \nleq \phi_{\IS}(x)$ with $\phi_{\IS}(x) \in \Lt_r$, we deduce that $r 
\not\models x \imp y$, which concludes the proof.
\end{proof}

\begin{theorem} \label{thm:rep-ref-axioms}
The extended Armstrong axioms are sound and complete provided comparability functions are 
reflexive. 
\end{theorem}

In spite of Theorem \ref{thm:rep-axioms}, the work of Day \cite{day1992lattice} on 
lattice implications still allows to reason on abstract FDs independently of any 
relations on $\C_{\U}$.
Let $\IS, \IS'$ be two sets of abstract FDs.
We say that $\IS, \IS'$ are \emph{equivalent} if $\Lt_{\IS} = \Lt_{\IS'}$.
Putting $\IS \vdash \IS'$ if $\IS \vdash x \imp y$ for every $x \imp y \in \IS'$, 
we immediately derive:

\begin{proposition}
The sets of abstract FDs $\IS_1$ and $\IS_2$ are equivalent if and only if $\IS_1 \vdash 
\IS_2$ and $\IS_2 \vdash \IS_1$.
Moreover, $\IS_1 \vdash \IS_2$ can be tested in polynomial time in the size of $\IS_1, 
\IS_2$ and $\C_{\U}$. 
\end{proposition}

A set $\IS$ of abstract FDs is \emph{redundant} if there exists an abstract FD $x \imp y$ 
in $\IS$ such that $\IS \setminus \{x \imp y\} \vdash x \imp y$. 
It is \emph{irredundant} otherwise. 
An abstract FD $x \imp y$ is \emph{right-closed} in $\IS$ if $y = \phi_{\IS}(x)$.
The \emph{left-saturation} of an abstract FD $x \imp y$ in $\IS$ is the abstract FD 
$\phi_{\IS'}(x) \imp y$ where $\IS' = \IS \setminus \{x \imp y\}$.
Eventually, $\IS$ is \emph{minimal} if for every equivalent set of abstract FDs $\IS'$, 
$\card{\IS} \leq \card{\IS'}$, where $\card{\IS}$ is the number of abstract FDs in $\IS$.
The Theorem 6.1 of Day \cite{day1992lattice} can be restated and adapted to our terms as 
follows:

\begin{theorem}[\cite{day1992lattice}] \label{thm:rep-minimal}
Let $\C_{\U}$ be a scheme context and $\IS$ be a set of abstract FD.
Consider the following three-step algorithm applied to $\IS$ and let $\IS_m$ be the 
resulting set of abstract FDs:
\begin{itemize}
\item right-closure
\item left-saturation
\item remove redundancy
\end{itemize}
Then, $\IS_m$ is minimal.
\end{theorem}

Observe that Theorem \ref{thm:rep-minimal} is the lattice counterpart of theorems on 
minimal covers for functional dependencies \cite{maier1983theory, wild1995computations}.
With the help of Proposition \ref{prop:rep-sigmaproof}, Algorithm \ref{alg:rep-closure}, 
and Theorem \ref{thm:rep-minimal}, we obtain

\begin{theorem}
Let $\C_{\U}$ be a scheme context, $\IS$ a set of abstract FDs over $\C_{\U}$.
The following tasks can be computed in poylnomial time in the size of $\C_{\U}$ and $\IS$:
\begin{itemize}
\item computing an equivalent irredundant set $\IS'$ of abstract FDs,
\item computing an equivalent minimal set $\IS$ of abstract FDs.
\end{itemize}
\end{theorem}

We conclude this section by giving hardness results regarding abstract FDs.
On this purpose, we show how abstract FDs generalize classical FDs.
Beforehand, we define some more notions.
A set of abstract FDs $\IS$ is an \emph{abstract cover} for a relation $r$ if for every 
abstract FD $x \imp y$, $r \models x \imp y$ if and only if 
$\IS\models x \imp y$.
An abstract tuple $x$ is an \emph{abstract key} of $r$ if $\phi_r(x) = 1_{\U}$, $\phi_r$ 
being the closure operator associated to $r$.
The abstract key $x$ is minimal if for every $y < x$ in $\Lt_r$, $y$ is not an abstract 
key of $r$.
The size of an abstract key is its number of non-bottom elements.

The subsequent discussion is a more formal version of the end of the paragraph 
\textit{Relational model (without nulls)} of Subsection 4.2. in the revised manuscript.
Let $\U$ be a relation scheme such that $\card{\U} = n$, and $\ctt{null} \notin \dom(A)$ 
for each $A \in \U$.
Let $\mathbf{F}$ be a set of functional dependencies over $\U$.
For each $A \in \U$, let $\Lt_A$ be the lattice $0 \leq 1$ and let $f_A$ be the 
comparability function given by $f(u, v) = 1$ if and only if $u = v$.
Let $\C_A = (A, f_A, \Lt_A)$ be the resulting attribute context and let $\C_{\U} = \{\C_A 
\mid A \in \U\}$ be the associated scheme context.
An abstract tuple $x$ is then a binary word over $\{0, 1\}^n$, which can be interpreted 
as the characteristic vector of the corresponding subset $X$ of $\U$.
The context $\C_{\U}$ admits a unique (strong) reality defined by $x_g = 1_{\U}$, where 
$1_{\U}$ is the top element of $\Lt_{\U}$.
The corresponding reality $g$ precisely coincides with the classical equality.
On the other hand, since the \ctt{null} value is not considered, the comparability 
functions are reflexive.
Thus, Theorem \ref{thm:rep-axioms} applies and $\IS \models x \imp y$ can be checked in 
polynomial time for every set of abstract FDs $\IS$ and every abstract FD $x \imp y$.
This lays the ground for the study of classical problems regarding abstract keys and 
abstract covers of a relation. 

The following proposition settles the relationship between abstract FDs and FDs in this 
context.

\begin{proposition}
We have $r \models X \imp Y$ if and only if $r \models x \imp y$.
As a consequence, we have:
\begin{itemize}
\item an abstract tuple $x$ is an abstract minimal key of $r$ if and only if the 
corresponding set $X$ is a minimal key of $r$,
\item a set of abstract FDs $\IS$ is an abstract cover of $r$ if and only if the 
corresponding set of FDs $\mathbf{F}$ is a cover of $r$.
\end{itemize}
\end{proposition}

\begin{proof}
This is a consequence of Theorem \ref{thm:projection-g} in the paper and the fact that 
the unique reality $g$ coincide with strict equality. 
\end{proof}

We derive the following more general theorem, which applies to any scheme context.

\begin{theorem}
Let $\C_{\U}$ be any scheme context.
Then:
\begin{itemize}
\item the problem of computing the size of a minimal abstract key of a relation is 
\NP-complete,
\item the problem of listing the abstract minimal keys of a relation is harder than the 
problem of listing the minimal keys of a relation.
\item the problem of computing an abstract minimal cover of a relation is harder than the 
problem of computing a minimal cover of a relation.
\end{itemize}
\end{theorem}

\end{document}